\documentclass[11pt]{article}

\usepackage{layout}
\usepackage{amsmath, amsthm}
\usepackage{amssymb, amsrefs}
\usepackage{graphicx}
\usepackage{arydshln}

\newtheorem{theorem}{Theorem}[section]

\newtheorem{proposition}[theorem]{Proposition}

\newtheorem{definition}[theorem]{Definition}

\begin{document}
\pagestyle{headings}

 \title{\sffamily On the Concrete Categories of Graphs}
    \author{
    \textsc{George McRae, Demitri Plessas, Liam Rafferty} \\[0.25em]
    }

\maketitle

\vspace*{-2em}

\begin{abstract}
\noindent
\par{In the standard Category of Graphs, the graphs allow only one edge to be incident to any two vertices, not necessarily distinct, and the graph morphisms must map edges to edges and vertices to vertices while preserving incidence. We refer to these graph morphisms as Strict Morphisms. We relax the condition on the graphs allowing any number of edges to be incident to any two vertices, as well as relaxing the condition on graph morphisms by allowing edges to be mapped to vertices, provided that incidence is still preserved. We call this broader graph category The Category of Conceptual Graphs, and define four other graph categories created by combinations of restrictions of the graph morphisms as well as restrictions on the allowed graphs.}
\par{We investigate which Lawvere axioms for the category of Sets and Functions apply to each of these Categories of Graphs, as well as the other categorial constructions of free objects, projective objects, generators, and their categorial duals.}
\end{abstract}

\tableofcontents

\section{Introduction - Conceptual Graphs and Their Morphisms}
\indent Often the study of morphisms of any mathematical object starts with the study of automorphisms. In graph theory, this study produced representation theorems for groups as automorphism groups of graphs. The study of automorphisms also produced characterizations of graphs (e.g. vertex-transitive graphs and distance-transitive graphs). More recently, finding restrictions of the automorphism set for the graph by considering automorphisms that fix vertex colors or by considering automorphisms that fix certain sets of vertices produced useful new graph parameters (the distinguishing number \cite{Albertson} and fixing number \cite{Courtney} for a graph).\\
\indent The study of morphisms then delves into the study of endomorphisms and finally homomorphisms. In graph theory, a certain class of graph homomorphisms generalize vertex-coloring, and are now being widely studied. In 2004, a textbook was published about these graph homomorphisms \cite{HN2004}.\\
\indent It is natural in the study of morphisms to use category theory, and our goal in this paper is to create the categorial framework to be used in further study of graph homomorphisms. Most of the proofs found in this paper will be combinatorial in flavor and many of the results lend themselves to combinatorial enumeration.\\
\indent The most common category considered in (undirected) graph theory is a category where graphs are defined as having at most one edge incident to any two vertices and at most one loop incident to any vertex. The morphisms are usually described as a pair of functions between the vertex sets and edge sets that respect edge incidence.\\
\indent In this paper, we will be relaxing these restrictions an investigating the concrete categories that are created. The rest of this section is concerned with defining our graphs as well as defining the restrictions. The second section defines six different categories of graphs, and proves that they are indeed categories. The third section investigates which of the six axioms for \textbf{Sets}, the category of sets and functions, apply to each of these categories of graphs. The last section investigates the other categorial constructions of free objects, projective objects, and generators, as well as their duals, and gives complete categorizations of these objects in the six categories of graphs.\\
\indent In our graphs, we want to start out with as great a generality as possible and add restrictions later. This means we want to allow graphs to have multiple edges between any two vertices and multiple loops at any vertex. We will define our graphs in the style of Bondy and Murty \cite{Bondy}, namely, graphs are sets of two kinds of parts: ``edges'' and ``vertices'' together with an ``incidence'' function. 
\begin{definition}
A \emph{conceptual graph} $G$ consists of\\
$G=\langle P(G), V(G); \partial_{G}:P(G) \rightarrow V(G)$\rotatebox[origin=c]{90}{$\ltimes$}$V(G), \iota_{G}:V(G)\hookrightarrow P(G) \rangle$ where $P(G)$ is the set of parts of $G$, $V(G)$ is the set of vertices of $G$, $V(G)$\rotatebox[origin=c]{90}{$\ltimes$}$V(G)$ is the set of unordered pairs of vertices of $G$, $\partial_{G}$ is the incidence map from the set of parts to the unordered pairs of vertices, $\iota_{G}$ is the inclusion map of the vertex set into the part set, and for $\underline{\Delta}:V(G) \rightarrow V(G)$\rotatebox[origin=c]{90}{$\ltimes$}$V(G)$, the unordered diagonal map, $\partial_G\iota_G=\underline{\Delta}$.
\end{definition}
\indent We define the set of edges of a graph, $G$, to be $E(G)=P(G)\backslash \iota_G(V(G))$. Henceforth, we will frequently abbreviate conceptual graph to graph. In our study here, we have no need to restrict our edge sets and vertex sets of our graphs to be finite sets.\\
\indent In \cite{Bondy} a graph does not have the inclusion map, $\iota$, but such a map will be critical when defining a graph homomorphism. In this way, we can think of the vertex ``part'' of the graph and the edge ``part'' of the graph in the same ``part'' set. We do allow $G=\emptyset$, i.e. $P(G)=\emptyset$, the empty graph, to be considered a graph. However, since $\partial_G$ is required to be a function, if $V(G)=\emptyset$ then $P(G)=\emptyset$. We also allow $V(G)\neq \emptyset$ and $E(G)=\emptyset$ (``no edges''), i.e. $P(G)=V(G)$.\\
\indent We now note the following. First, we naturally use the topologist's ``boundary'' symbol for incidence. Second, an unordered pair in $V(G)$\rotatebox[origin=c]{90}{$\ltimes$}$V(G)$ is denoted $u$\textunderscore$v$ or $(u$\textunderscore$v)$, for vertices $u,v\in V(G)$. Thus the natural unordered diagonal map $\underline{\Delta}:V(G)\rightarrow V(G)$\rotatebox[origin=c]{90}{$\ltimes$}$V(G)$ is given by $\underline{\Delta}(v)=v$\textunderscore$v$ or $(v$\textunderscore$v)$. Finally, we have chosen to consider our vertex set and edge set to be combined into a ``part'' set. Thus as an abstract data structure our graphs are a pair of sets: a set of parts with a distinguished subset called ``vertices''. This is done to make the description of morphisms more natural, i.e. functions between the ``over'' sets (of parts) that takes the distinguished subset to the other distinguished subset. This is similar to the construction of the Category of Topological Pairs of Spaces: for example, an object $(X,A)$ is a topological space $X$ with a subspace $A$ and a morphism $f:(X,A)\rightarrow(Y,B)$ is a continuous function from the topological space $X$ to the topological space $Y$ with $f[A]\subseteq f[B]$.\\
\indent We now define our morphisms for conceptual graphs.
\begin{definition}
$f:G\rightarrow H$ is a \emph{graph (homo)morphism} of conceptual graphs from $G$ to $H$ if $f$ is a function $f_{P}:P(G) \rightarrow P(H)$ and $f_{V}=f_P|_{V(G)}:V(G) \rightarrow V(H)$ that preserves incidence, i.e. $\partial_{H}(f_{P}(e))=(f_{V}(x)$\textunderscore$f_{V}(y))$ whenever $\partial_{G}(e)=(x$\textunderscore$y)$, for all $e\in P(G)$ and some $x,y \in V(G)$.
\end{definition}
\begin{figure}[h]
\centering \includegraphics[scale=.6]{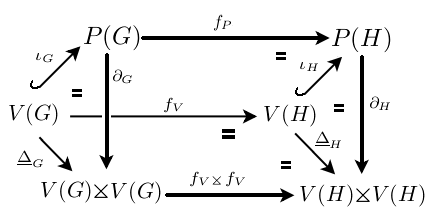}
\caption{The Graph Morphism}
\end{figure}
\indent This definition allows a graph homomorphism to map an edge to a vertex as long as the incidence of the edges are preserved. As an edge, $e\in E(G)$, can be mapped to the part set of the co-domain graph, $H$, so that it is the image of a vertex, i.e. $f(e)=\iota_H(v)$ for some $v\in V(H)$.\\
\indent We will now define some specialized classes of graphs and a specialized graph morphism. Our first restriction is a common restriction in graph theory. The set of graphs is restricted to allow only one edge between any two vertices (see \cite{HN2004}), and at most one edge between a vertex and itself (a loop). We call these graphs \textit{simple graphs} and define them in terms of conceptual graphs.
\begin{definition}
A \emph{simple} graph $G$ is a conceptual graph such that for all $u,v\in V(G)$ with $u \neq v$, there is at most one $e\in P(G)$ such that $\partial_G(e)=(u$\textunderscore$v)$, and for all $w \in V(G)$ there is at most one $f\in E(G)=P(G)\backslash\iota_G(V(G))$ such that $\partial_G(f)=(w$\textunderscore$w)$ (where $(u$\textunderscore$v)$ is the unordered pair of vertices $u$ and $v$).
\end{definition}
Thus, a graph is simple if and only if the incidence map is injective.\\
\indent Another common restriction is to not allow loops at all. This restriction is often required when discussing vertex coloring. We call these graphs \textit{loopless graphs}.
\begin{definition}
A \emph{loopless} graph $G$ is a conceptual graph such that for all vertices $u \in V(G)$ there is no edge $e\in E(G)=P(G)\backslash\iota_G(V(G))$ such that $\partial_G(e)=(u$\textunderscore$u)$.
\end{definition}
\indent This is not the usual notion of a ``simple'' graph often common in graph theory; that notion is the simple and loopless graph by our definition. This is a departure from standard nomenclature, but it fits our categorial discussion best.\\
\indent We now define the most common notion for a graph morphism in literature, we call it a \emph{strict} morphism because it always takes an edge part to a strict edge part (and not just a part, e.g. a vertex). The following definition is a modified form of the definition presented in \cite{HN2004} to apply to conceptual graphs.
\begin{definition}
Let $G$ and $H$ be conceptual graphs. A \emph{strict graph homomorphism} (or \emph{strict morphism}) $f:G\rightarrow H$ is a graph morphism $f:G\rightarrow H$ such that the strict edge condition holds: for all edges $e\in E(G)$, $f_P(e)\in E(H)$, i.e. the image under the strict morphism $f$ of an edge is again and edge. 
\end{definition}
\indent The condition, $\partial_{H}(f_{P}(e))=(f_{V}(x)$\textunderscore$f_{V}(y))$ whenever $\partial_{G}(e)=(x$\textunderscore$y)$, assures that the incidence of the edges in $G$ is preserved in $H$ under $f$. Note that the above definition also requires that vertices be mapped to vertices and edges be mapped (strictly) to edges. However, sometimes it may be beneficial to allow edges to be mapped to vertices. Such a morphism would allow a graph to naturally map to the contraction or quotient graph obtained by the contraction of an edge, but this could not be a strict morphism.\\
\indent Now that we have defined our graphs and graph homomorphisms, we are ready to discuss the various Categories of Graphs.

\section{The Categories of Graphs}
\begin{definition} \emph{Concerete Categories} \cite{Mac} are categories whose objects are sets with structure and whose morphisms are functions that preserve that structure.
\end{definition}
\indent We will now define six concrete categories of graphs using the various restrictions of the previous section. We do not include all combinations of restrictions, but instead focus on the combinations of restrictions often seen in literature.
\begin{definition}
The \emph{Category of Conceptual Graphs}, \textbf{Grphs}, is a (concrete) category where the objects are conceptual graphs and the morphisms are graph morphisms.
\end{definition}
\indent Keith Ken Williams \cite{KKWil} proved that the axioms of a category are satisfied by this definition. 
\begin{proposition}
\textbf{Grphs} is a category.
\end{proposition}
\indent This category, \textbf{Grphs}, we will think of as the big ``mother'' category of graphs. We now define five other commonly studied concrete subcategories of \textbf{Grphs}.
\begin{definition}
The \emph{Category of Simple Graphs with Conceptual Morphisms}, \textbf{SiGrphs}, is the (concrete) category where the objects are simple graphs, and the morphisms are conceptual graph morphisms.
\end{definition}
\begin{definition}
The \emph{Category of Simple Loopless Graphs with Conceptual Morphisms}, \textbf{SiLlGrphs}, is the (concrete) category where the objects are simple graphs without loops, and the morphisms are conceptual graph morphisms.
\end{definition}
\begin{definition}
The \emph{Category of Conceptual Graphs with Strict Morphisms}, \textbf{StGrphs}, is the (concrete) category where the objects are conceptual graphs, and the morphisms are strict graph morphisms.
\end{definition}
\begin{definition}
The \emph{Category of Simple Graphs with Strict Morphisms}, \textbf{SiStGrphs}, is the (concrete) category where the objects are simple graphs and the morphisms are strict graph morphisms.
\end{definition}
\indent This last category we defined is most often referred to as the ``category of graphs'' and is the main category of graphs discussed in \cites{HN2004, Dochtermann}, namely, graphs with at most one edge between vertices, at most one loop at a vertex, and all the morphisms are \emph{strict} (i.e. take and edge or loop \emph{strictly} to an edge or loop).
\begin{definition}
The \emph{Category of Simple Loopless Graphs with Strict Morphisms}, \textbf{SiLlStGrphs}, is the (concrete) category where the objects are simple graphs without loops, and the morphisms are strict graph morphisms.
\end{definition}
As the composition of strict morphisms are strict morphisms, and the identity morphism is a strict morphism, these are in fact categories. We now have a containment picture of our six different (concrete) categories of graphs, with our mother category at the top. At the bottom, we have also included as another graph category, the category of sets (and functions), where a set is considered as a graph with no edges and any function is a (strict) morphism of such graphs.
\begin{figure}[h]
\centering \includegraphics[scale=.6]{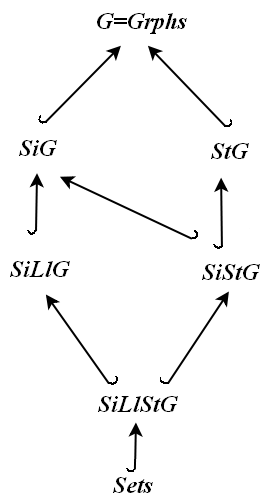}
\caption{The Categories of Graphs}
\end{figure}

\section[Categorial Comparisons - Graphs vs. Sets]{Categorial Comparisons - The Categories of Graphs vs. The Category of Sets}
\subsection{The Lawvere Axioms for \textbf{Sets}}
A natural representation question is, what characterizing properties must an abstract category have in order for it to be \textbf{Sets} (up to functor equivalence of categories)? There are six characterizing conditions for an arbitrary category to be the category \textbf{Sets} called the Lawvere-Tierney Axioms \cites{LT, Lawvere, T}. An arbitrary category is the category \textbf{Sets} if and only if all six of the following conditions are satisfied.\par
\begin{list}{}
\item \textbf{(L1)} \textbf{Sets} has Limits (and Colimits):
\item \textbf{(L2)} \textbf{Sets} has exponentiation with evaluation.
\item \textbf{(L3)} \textbf{Sets} has a subobject classifier.
\item \textbf{(L4)} \textbf{Sets} has a natural number object.
\item \textbf{(L5)} \textbf{Sets} has the Axiom of Choice.
\item \textbf{(L6)} The subobject classifier in \textbf{Sets} is two-valued.
\end{list}

Each of these categorial properties are defined abstractly \cite{Goldblatt}. That is to say, in the definitions only objects and morphisms will be used, not the structure of the objects.\par
We investigate each of our categories of graphs to see which of the axioms for \textbf{Sets} they satisfy and which axioms they do not. We start with \textbf{Grphs}. It has already been shown that \textbf{Grphs} satisfies (L1) \cite{KKWil}, and that \textbf{Grphs} fails to have the (L2) construction \cite{BMS}. We will provide the constructions for (L1) for completeness and provide an alternate proof for the failure of the existence of (L2) that has a combinatorial flavor, shows that an exponentiation object can exist, and shows exponentiation with evaluation fails due to an evaluation morphism that fails to satisfy the universal mapping property (as opposed to \cite{BMS} who show the failure of a necessary adjoint relationship).\par
\subsection{Lawvere-type Axioms for \textbf{Grphs}}
\begin{proposition}
\textbf{Grphs} satisfies axioms (L1), (L3), and (L4) and does not satisfy axioms (L2), (L5), and (L6).
\end{proposition}
\begin{proof}
\textbf{Axiom L1 (``limits'') holds for} \textbf{Grphs}: We note that limits and colimits exist by the existence of a terminal object, products, equalizers, and their duals \cite{Mac}. The one vertex graph $K_1$ (the classical ``complete graph on one vertex'') is the terminal object, which we will denote $\hat{1}$, and the empty vertex set (and edge set) graph $\emptyset$ is the initial object, which we will denote $\hat{0}$.\\
\indent For products, let $A$ and $B$ be graphs in \textbf{Grphs}. To keep track of the incidence relation we will need a total ordering (arbitrary but fixed) on the vertex sets. Let $\leq_A$ be the total ordering on $V(A)$ and $\leq_B$ be the total ordering on $V(B)$. For any element $\alpha \in P(A)$ with $\partial_A(\alpha)=(a_1$\textunderscore$a_2)$ and $a_1\leq_A a_2$, and any element $\beta \in P(B)$ with $\partial_B(\beta)=(b_1$\textunderscore$b_2)$ and $b_1\leq_B b_2$, there is an element $(\alpha,\beta)$ in $P(A\times B)$ whose incidence is $\partial_{A\times B}((\alpha,\beta ))=((a_1,b_1)$\textunderscore$(a_2,b_2))$, (hence $V(A\times B)=V(A)\times V(B)$), and if $a_1\neq a_2$ and $b_1\neq b_2$ then there is an additional element in $P(A\times B)$, $\overline{(\alpha,\beta)}$, whose incidence is $\partial(\overline{(\alpha,\beta)})=((a_2,b_1)$\textunderscore$(a_1,b_2))$. The projections from the product are $\pi_A((\alpha,\beta))=\alpha$ (regardless of the bar,  $\pi_A(\overline{(\alpha,\beta)})=\alpha$) and $\pi_B((\alpha,\beta))=\beta$ (regardless of the bar,  $\pi_B(\overline{(\alpha,\beta)})=\beta$). So we have a vertex set, a part set, and an incidence relation from the edge set to an unordered product of the vertex set with itself; therefore this product is in fact, a graph. It can easily be shown it satisfies the universal mapping properties for a product.\par
\textbf{Remark}: for our figures with graphs, we provide ``pictures'' (with picture frames) for the graphs. This helps to distinguish the graphs from the morphisms, especially in the case of graphs with multiple components. It also emphasizes that we are often choosing representative graphs from an isomorphism class of graphs.\\
\begin{figure}[h]
\centering \includegraphics[scale=.7]{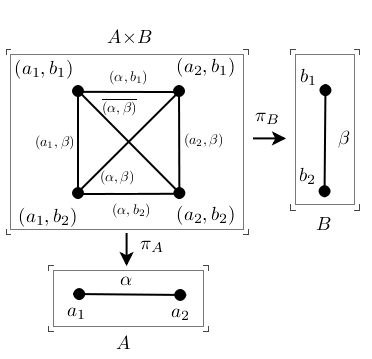}
\caption{An example of the product in \textbf{Grphs} with pictures.}
\end{figure}\par
\indent The coproduct of two graphs in \textbf{Grphs} is the disjoint union of the two graphs. This construction satisfies the universal mapping property for a coproduct.\\
\indent Let $f,g:A\rightarrow B$ be two morphisms in \textbf{Grphs}. We define the equalizer's part set as $P(Eq)=\{a \in P(A) \vert f(a)=g(a)$ and if $\partial(a)=(u$\textunderscore$v)$ then $f(u)=g(u)$ and $f(v)=g(v)\}$. We have required that a edge is in the part set of the equalizer implies that its incident vertices are in the part set of the equalizer. So the part set of the equalizer is a subset of $P(A)$ and the vertex set of the equalizer is a subset of $V(A)$. And by construction the incidence function from $A$ can be restricted to the part set of the equalizer since we made sure to only include edges if their incident vertices were in the equalizer. So we have that the equalizer is a graph and that $\iota:Eq\hookrightarrow A$ is just the inclusion morphism. This construction satisfies the universal mapping property for an equalizer.\\
\indent Again, let $f,g:A\rightarrow B$ be two morphisms in \textbf{Grphs}.
In order to construct the coequalizer, $Q$, we will need to mod out $P(B)$ by an equivalence relation. Let $x,y \in P(B)$, define $x \sim y$ if there exists a sequence $a_0,a_1,...,a_n\in P(A)$ such that $x=f(a_0)$, $g(a_0)=f(a_1)$, $g(a_1)=f(a_2)$,...., $g(a_{n-1})=f(a_n)$ and $y=f(a_n)$ or $y=g(a_n)$ (or switch the $f$'s and $g$'s in the previous statement). \par
When we identify equivalent elements of $P(B)$ we get a set, which we will call $P(Q)$, and we will define $q:P(B)\rightarrow P(Q)$ as the set function which maps elements to their equivalence class in $P(Q)$.\\
\indent We now show $Q$ is a graph. First we note that if an edge or vertex is identified with another edge or vertex, then their incident vertices must be identified. This is due to the fact that if $f(e_1)=g(e_2)$ with $\partial(e_1)=(a$\textunderscore$b)$ and $\partial(e_2)=(c$\textunderscore$d)$ by the properties of graph morphisms $f(a)=g(c)$ and $f(b)=g(d)$ or $f(a)=g(d)$ and $f(b)=g(c)$. Either way we can construct a sequence between the vertices of $e_1$ and $e_2$ of their incidences which will be the sequence necessary for the equivalence of their vertices. It is worth noting that some equivalence classes in \textbf{Grphs} have both vertices and edges in them, and if an equivalence class contains any vertices, it is a vertex in the coequalizer graph. Hence by construction $q$ is a graph morphism and $qf=qg$. This construction satisfies the universal mapping property for a coequalizer.\\
\indent \textbf{Axiom L2 (``exponentiation with evaluation'') fails for} \textbf{Grphs}: By way of contradiction let us suppose that categorial exponentiation with evaluation (by the universal mapping property of exponentiation with evaluation) exists in \textbf{Grphs}. Then \textbf{Grphs} has a terminal object, products and exponentiation with evaluation. Hence there is an adjoint functor relationship between $hom_{\mathbf{Grphs}}(X\times A, B)$ and $hom_{\mathbf{Grphs}}(X,B^A)$ for all graphs $A,B$ and $X$. Hence there is a bijection between the set of morphisms $X\times A \rightarrow B$ and the set of morphisms $X \rightarrow B^A$.\par
We construct the counterexample to the existence of exponentiation and evaluation in \textbf{Grphs} in two steps. First, we use the adjoint functor relationship to completely determine (by a ``brute force'' count) the vertices, edges, and incidence of $B^A$ as a graph where $B=K_1^\ell$, the graph with a single vertex with a loop on the vertex, and $A=K_2$, the classical ``complete graph on two vertices''. Second, we show that for all morphisms $B^A\times A\rightarrow B$ that satisfy the commuting morphism equations of the evaluation universal mapping property fails to have the uniqueness requirement for the universal mapping property for exponentiation with evaluation.\par
To begin, we will use the above mentioned adjoint bijection, but for various choices of ``test'' objects $X$. First, for $X=\hat{1}$ = terminal object = a single vertex graph, $X\times A\cong A$.\par
Any morphism from $X\times A$ to $B$ must send both vertices to the single vertex in $B$, the edge may be sent to either the loop or to the vertex. So there are two maps here. Therefore, there must be two morphisms from $\hat{1}$ to $B^A$. Since $\hat{1}$ is just a single vertex, $B^A$ must have exactly two vertices.\par
Second, suppose $X=K_1^{\ell}$ is a vertex with a single loop. Then $X\times A$ is a graph on two vertices, with a loop at each vertex, and two edges incident to the two distinct vertices.\par
Again, both the vertices of $X\times A$ must be sent to the single vertex of $B$. Now there are four edges in $X\times A$, each edge maybe sent to either the loop or to the vertex (independent of where the other edges are sent). So there are $2^4=16$ morphisms here. Therefore there must be exactly $16$ morphisms from $X$ to $B^A$. There are exactly two morphisms which send both the edge and the vertex of $X$ to a single vertex (since we have already determined that there are only two vertices in $B^A$). Which leaves 14 more morphisms to account for. Since the vertex of $X$ must be sent to a vertex, and the loop must be either sent to a loop or a vertex, we conclude that there are $14$ loops distributed between the two vertices (we do not know how they are divided between the two, but we know that there are $14$ of them).\par
Third, suppose $X=K_2$ is two vertices with a non-loop edge between them. Then $X\times A \cong K_4$\par
Again, all four vertices of $X\times A$ must be sent to the single vertex of $B$. The six edges of $X\times A$ can be sent to either the loop or to the vertex (independent of where the other ones are sent). So there are $2^6=64$ morphisms here.\par
Therefore there must be $64$ morphisms from $X$ to $B^A$. $X$ has only one edge, it can either be sent to a vertex, a loop, or a non-loop edge. There are two ways to send it to a vertex (and this will force both its vertices to be sent to this vertex to preserve incidence). Since there are $14$ loops, there are $14$ ways to send it to a loop (and since incidence must be preserved both the vertices of $X$ must be sent to the vertex incident on this loop, it is worth noting that we still don't know where these $14$ loops are, but it doesn't matter counting these morphisms). Which leaves us with $48$ morphisms to account for, which must send the edge of $X$ to a non-loop edge of $B^A$ since we have accounted for the other possibilities. There are only two vertices in $B^A$ so there is only place to send a non-loop edge. Also, each non-loop edge in $B^A$ will give us two morphisms from $X$ to that edge (once you decide which of the vertices to send one vertex of $X$ to, the edge must be sent to the  edge and the other vertex of $X$ to the other vertex of $B^A$). Therefore there must be precisely $24$ non-loop edges connecting the two vertices of $B^A$. We still do not know where the $14$ loops are in $B^A$ but we have a pretty good idea of what it must look like.\par
Now we will test what $B^A$ must be by testing with one more $X$ to determine the placement of the loops. So for the fourth test choice of $X$, suppose $X$ is two vertices with one loop and one non-loop edge.\par
\begin{figure}[h]
\centering \includegraphics[scale=.6]{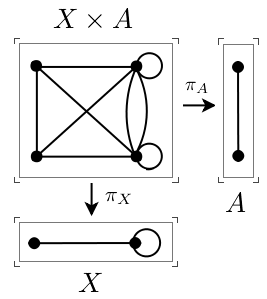}
\caption{A picture of $X\times A$, for the fourth test choice of $X$.}
\end{figure}\par
Again, all four vertices of $X\times A$ must be sent to the single vertex in $A$, and each of the 9 edges can either be sent to the loop or to the vertex (independent of where the other edges go). So there are $2^{9}=512$ morphisms here.\par
Now we will count the number of morphisms $X\rightarrow B^A$ by considering the following six disjoint types of morphisms whose union are all the morphisms: everything in $X$ can be sent to a single vertex, everything in $X$ but the loop can be sent to a vertex with the loop to a loop, the non-loop edge can be sent to a loop with everything else sent to a vertex, the non-loop edge can be sent to a non-loop edge with the loop sent to a vertex, the loop can be sent to a loop and the non-loop edge sent to a non-loop edge, or both the loop and the non-loop edge can be sent to loops.\par
There are two ways to send everything in $X$ to a vertex since $B^A$ only has two vertices. There are $14$ ways to send the loop of $X$ to a loop and everything else to a vertex since $B^A$ has $14$ loops. Likewise there are $14$ ways to send the non-loop edge to a loop with everything else going to the incident vertex. As discussed before, there are $48$ ways to send the non-loop edge to a non-loop edge and the loop to a vertex.\par
Now we will count the number of ways to send the loop of $X$ to a loop in $B^A$ and the non-loop edge of $X$ to a non-loop edge of $B^A$. There are $14$ choices of where to send the loop, and this choice determines where vertex incident on the loop is sent. After this choice is made, there will be $24$ non-loop edges in $B^A$ to send the non-loop edge of $X$ to (note again that we don't know which vertex the loops are on, but it does not effect our count of this type of morphism). So there are $14\times 24=336$ of this type of morphism.\par
We have now accounted for $2+14+14+48+336=414$ morphisms, which leaves $512-414=98$ morphisms to account for. The only other type of morphism is one which sends both the loop and the non-loop edge of $X$ to loops in $B^A$. Suppose there are $m$ loops on one vertex of $B^A$ and $n$ loops on the other. Then there is $m^2+n^2=98$ morphisms which send both edges of $X$ to a loop, and $m+n=14$. Solving this system of equations yields the unique solution of $m=7$ and $n=7$. Hence the $14$ loops are distributed evenly between the two vertices of $B^A$.\par
So we now have a complete description of what we will call the ``exponential object'' $B^A$, for the given $A$ and given $B$. (This assumed that categorial exponentiation with evaluation exists).\par
We've determined that $B^A$ is a graph with two vertices (which we will label $u$ and $v$) $24$ non-loop edges (which we will label $e_i$ for $i=1, \dots, 24$), and 7 loops on each vertex (which we will label $u_{\ell_j}$ and $v_{\ell_j}$ for $j=1,\dots,7$). It will also help us to label the graphs $A$ and $B$. Label the vertices of $A\cong K_2$ as $a_1$ and $a_2$, and the edge as $e_a$. Label the vertex of $B$ by $b$ and the loop by $\ell_b$.\par
\begin{figure}[h]
\centering\includegraphics[scale=.6]{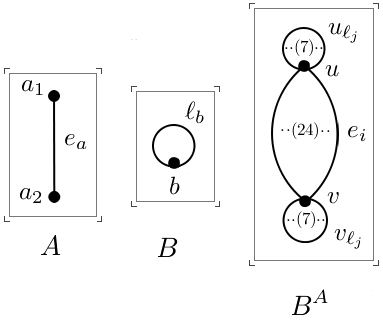}
\caption{Pictures of the graphs for the counterexample to categorial ``exponentiation'' in \textbf{Grphs}.}
\end{figure}
\par
But even though this exponential object exists, we have yet to show that its evaluation satisfies the uniqueness feature of the universal mapping property for exponentiation with evaluation. So we investigate evaluation by again using test objects in the universal mapping property for exponentiation with evaluation, which states that there exists $ev:B^A\times A \rightarrow B$ such that for all $X$ and $g:X\rightarrow B$, there is a unique $\overline{g}:X\rightarrow B^A$ such that $g=ev(\overline{g}\times 1_A)$.\par
For the first choice of test objects, let $X$ be the single vertex graph $K_1$ (which we have denoted $\hat{1}$) with vertex $x$. Then as $X(=\hat{1})$ is the terminal object, $X\times A\cong A$. Thus there are two morphisms from $X\times A$ to $B$. Let $g_1:X\times A \rightarrow B$ be the morphism which maps all of $P(X\times A)$ to the vertex $b$ of $B$, and let $g_2:X\times A \rightarrow B$ by the morphism that maps the edge of $X \times A$ to the loop $\ell_b$ of $B$.\par
Consider $g_1$, by the universal mapping property there is a unique $\overline{g}:X\rightarrow B^A$ such that $ev(\overline{g}\times 1_A)=g$. There are two morphisms from $X$ to $B^A$, $\overline{g}(x)=u$ or $\overline{g}(x)=v$. If $\overline{g}(x)=u$, then $ev(\overline{g}\times 1_A((x,e_a)))=ev((u,e_a))=g(x)=b$. If $\overline{g}(x)=v$, then $ev(\overline{g}\times 1_A((x,e_a)))=ev((v,e_a))=g(x)=b$.\par
As $\overline{g}$ is unique, only one of the two above possibilities holds. Since there is an automorphism of $B^A\times A$ that exchanges $(u,e_a)$ and $(v,e_a)$ (exchange all labels of $u$ and $v$ and swap $(e_i,e_a)$ with $\overline{(e_i,e_a)}$), without loss of generality we can choose $ev((v,e_a))=b$. Then as $\overline{g}$ is unique, $ev((v,e_a))\neq ev((u,e_a))$. Hence $ev((u,e_a))=\ell_b$.\par
For the second choice of test objects, test with $X=A (=K_2)$ to achieve a contradiction.\par
We claim that $g((a_1,e_a))=b$ if and only if $\overline{g}(a_1)=v$ and $g((a_1,e_a))=\ell_b$ if and only if $\overline{g}(a_1)=u$. For, since $ev((v,e_a))=b$ and $ev((u,e_a))=\ell_b$, $g((a_1,e_a))=ev(\overline{g}\times 1_A((a_1,e_a)))=ev((\overline{g}(a_1),e_a))$. Hence $g((a_1,e_a))=b$ if and only if $\overline{g}(a_1)=v$ and $g((a_1,e_a))=\ell_b$ if and only if $\overline{g}(a_1)=u$.\par
A similar argument shows $g((a_2,e_a))=b$ if and only if $\overline{g}(a_2)=v$ and $g((a_2,e_a))=\ell_b$ if and only if $\overline{g}(a_2)=u$.\par
Then for $g:A\times A \rightarrow B$ with $g((a_1,e_a))=b$ and $g((a_2,e_a))=\ell_b$, we have $\overline{g}(a_1)=v$ and $\overline{g}(a_2)=u$. Hence for such a $g$, as $\overline{g}$ must preserve incidence, $\overline{g}(e_a)=e_i$ for some $i=1,\dots, 24$. We now notice the following useful observation.\par
\textbf{(1)} If $\overline{g}(e_x)=e_i$ for some $i=1,\dots, 24$ then $\overline{g}\times 1_A((e_a,a_1))=(e_i,a_1)$, $\overline{g}\times 1_A((e_a,a_2))=(e_i,a_2)$, $\overline{g}\times 1_A((e_a,e_a))=(e_i,e_a)$, and $\overline{g}\times 1_A(\overline{(e_a,e_a)})=\overline{(e_i,e_a)}$.\par
For each $i=1,\dots,24$ there are two choices of where to map each of $(e_i,a_1)$, $(e_i,a_2)$, $(e_i,e_a)$, and $\overline{(e_i,e_a)}$ in a morphism from $B^A\times A\rightarrow B$ (either to $b$ or $\ell_b$). Thus for a fixed $i$, there are 16 possible ways to map the edges $(e_i,a_1)$, $(e_i,a_2)$, $(e_i,e_a)$, and $\overline{(e_i,e_a)}$ to $B$. However, there are 24 such indicies. Thus by the pigeonhole principle,\par
\textbf{(2)} there exists $i,j\in\{1,\dots,24\}$ with $i\neq j$, $ev((e_i,a_1))=ev((e_j,a_1))$, $ev((e_i,a_2))=ev((e_j,a_2))$, $ev((e_i,e_a))=ev((e_j,ea))$, and $ev(\overline{(e_i,e_a)})=ev(\overline{(e_j,e_a)})$.\par
So define a morphism $g:A\times A\rightarrow B$ by $g(x)=b$ for all vertices $b\in V(A\times A)$, $g((a_1,e_a))=b$, $g((a_2,e_a))=\ell_b$, $g((e_a,a_1))=ev((e_j,a_1))$, $g((e_a,a_2))=ev((e_j,a_2))$, $g((e_a,e_a))=ev((e_j,e_a))$, and $g(\overline{(e_a,e_a)})=ev(\overline{(e_j,e_a)})$ (incidence is trivially preserved). Then there is a unique $\overline{g}:A\rightarrow B^A$ such that $ev(\overline{g}\times 1_A)=g$.\par
However, by (1) $\overline{g}(a_1)=v$ and $\overline{g}(a_2)=u$, $\overline{g}(e_a)=e_j$ is such a morphism and by (2) $\overline{\overline{g}}(a_1)=v$, $\overline{\overline{g}}(a_2)=u$, and $\overline{\overline{g}}(e_a)=e_i$ is another. Hence no such unique morphism exists and (L2) does not hold in \textbf{Grphs}.\par

\textbf{Axiom L3 (``subobject classifier'') holds in} \textbf{Grphs}: In \textbf{Grphs} the subobject classifier is the following graph:\par
\begin{figure}[h]
\begin{center}
\includegraphics[scale=.6]{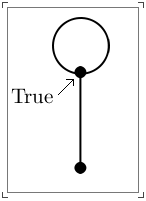}
\end{center}
\caption{A picture of the subobject classifier $\Omega$ in \textbf{Grphs}}
\end{figure}
together with the canonical morphism $\top:\hat{1}\rightarrow \Omega$ from the terminal object $\hat{1}$ to the subobject classifier $\Omega$, which maps the vertex of $\hat{1}$ to the vertex labeled ``True'' in the above picture.\par
For any subgraph $A$ of a graph $X$ we must map all of $A$ to this single true vertex. And we must map all the vertices of $V(X)\backslash V(A)$ to the other vertex. Any edge which is not in the image of $A$ but is incident only to vertices in the image of $A$ is mapped to the loop. Any edge with only one incident vertex in the image of $A$ is mapped to the non-loop edge. This construction satisfies the universal mapping property for subobject classifiers.\\
\indent \textbf{Axiom L4 (``natural number object'') holds in} \textbf{Grphs}: The natural number object in \textbf{Grphs}, $\mathbb{N}$, is the graph with no edges, and a countably infinite number of vertices labeled by the natural numbers, coupled with the initial morphism $\ulcorner 0 \urcorner : \hat{1} \rightarrow \mathbb{N}$ from the terminal object $\hat{1}$ defined by mapping the single vertex of the terminal object to the vertex labeled $0$, and successor function $\sigma: \mathbb{N} \rightarrow \mathbb{N}$ where given a vertex labeled $n$, $\sigma(n)=n+1$. This construction satisfies the universal mapping property for a natural number object.\\

\indent \textbf{Axiom L5 (``choice'') fails in} \textbf{Grphs}: Consider the graph morphism in figure 8 where $f(a_1)=b_1$ and $f(a_2)=b_2$.\par
\begin{figure}[h]
\begin{center}
\includegraphics[scale=.6]{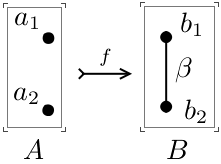}
\end{center}
\caption{A picture for the counterexample to (L5 - ``choice'') in \textbf{Grphs}}
\end{figure}
For any $g:B\rightarrow A$, $g$ must send the edge $\beta$ to one of the vertices. Without loss of generality, assume $g(\beta)=a_1$, then we must have $g(b_1)=a_1$ and $g(b_2)=a_1$ to preserve incidence. But then $fgf\neq f$. So we have an example where there does not exist a $g:B\rightarrow A$ such that $fgf=f$.\\
\indent \textbf{Axiom L6 (``two-valued'') fails for} \textbf{Grphs}: By applying the definition of terminal object and coproduct we have that $\hat{1}+\hat{1}$ is a two vertex graph with no edges. But $\Omega$ has edges so it can not be isomorphic to $\hat{1}+\hat{1}$ and is therefore not two-valued.
\end{proof}

\subsection{Lawvere-type Axioms for \textbf{SiGrphs}}
\begin{proposition}
\textbf{SiGrphs} satisfies axioms (L1), (L3), and (L4) and does not satisfy axioms (L2), (L5) and (L6).
\end{proposition}
\begin{proof}
\textbf{Axiom L1 (``limits'') holds for} \textbf{SiGrphs}: The proof of existence products and coproducts in \textbf{SiGrphs} follows similarly to the proof of existence of products and coproducts in \textbf{Grphs} using same constructions, by identifying any multiple edges that occur as a single edge and any multiple loops that occur as a single loop.\\
\indent \textbf{Axiom L2 (``exponentiation with evaluation'') fails for} \textbf{SiGrphs}: Suppose that exponentiation with evaluation exists in \textbf{SiGrphs}. Then \textbf{SiGrphs} has a terminal object, products and exponentiation with evaluation. Thus there is a standard adjoint functor relationship creating a bijection between the set of morphisms $X\times B \rightarrow A$ and the set of morphisms $X \rightarrow A^B$.\\
\indent To construct our counterexample to the existence of exponentiation with evaluation in \textbf{SiGrphs}, let both $A$ be $K_1^\ell$ the graph with a single vertex with a loop, and $B$ be $K_2$. To begin, we will let $X$ be a single vertex (this is the multiplicative identity, $\hat{1}$, in \textbf{SiGrphs} and hence $X\times B=B$). As $K_2$ admits 2 morphisms to $K_1^\ell$ in \textbf{SiGrphs}, then $A^B$ has 2 vertices, identified by $X$.\\
\indent Now let $X$ be $K_1^\ell$. Then $X\times B$ is a graph with two vertices with a loop at each vertex and an edge between the two vertices. $X\times B$ admits 8 morphism to $K_1^\ell$, as the edge and each loop can be mapped to either the vertex or the loop. Hence $K_1^\ell$ admits 8 morphisms to $A^B$. However, in \textbf{SiGrphs}, $K_1^\ell$ admits at most 4 morphisms to any graph on two vertices, 2 morphisms that map the loop to a vertex, and 2 that map the loop to another loop. Hence we have a contradiction and exponentiation with evaluation does not exist in \textbf{SiGrphs}.\\
\indent \textbf{Axioms L3 and L4 (``subobject classifier'' and ``natural number object'') both hold in} \textbf{SiGrphs}: Both the subobject classifier, and the natural number object for \textbf{SiGrphs} is the same as it is for \textbf{Grphs}.\\
\indent \textbf{Axioms L5 and L6 (``choice'' and ``two-valued'') fail for} \textbf{SiGrphs}: The same counterexample for \textbf{Grphs} applies here.
\end{proof}

\subsection{Lawvere-type Axioms for \textbf{SiLlGrphs}}
\begin{proposition}
\textbf{SiLlGrphs} satisfies axioms (L1), (L2), and (L4), and does not satisfy axioms (L3), (L5), and (L6).
\end{proposition}
\begin{proof}
\textbf{Axiom L1 (``limits'') holds for} \textbf{SiLlGrphs}: To show the existence of limits, we note that the terminal object, products, and equalizers are defined as in \textbf{SiGrphs}. For colimits, the initial object, and the coproduct are the same as in \textbf{SiGrphs} with the coequalizer being the construction given in \textbf{Grphs} with multiple edges identified as a single edge, and loops identified with the incident vertex.\\
\indent \textbf{Axiom L2 (``exponentiation with evaluation'') holds for} \textbf{SiLlGrphs}: Given graphs $G$ and $H$, define $H^G$ by $V(H^G)=$ hom$_\mathbf{SiLlG}(G,H)$, and $e\in P(H^G)$ with $\partial_{H^G}(e)=(f_1$\textunderscore$f_2)$ if for all $d\in P(G)$ with $\partial_G(d)=(d_1$\textunderscore$d_2)$, there exists $d^\prime\in P(H)$ with $\partial_H(d^\prime)=(f_1(d_1)$\textunderscore$f_2(d_2))$.\\
\indent Then define $ev:H^G\times G\rightarrow H$ by $ev((f,v))=f(v)$ for all vertices $(f,v)\in V(H^G\times G)$ and for $e\in P(H^G \times G)$ such that $\partial_{H^G \times G}(e)=((f,v)$\textunderscore$(g,u))$ define $ev(e)=d$ for $d\in P(H)$ with $\partial_H(d)=(f(v)$\textunderscore$g(u))$. Such a $d$ exists by construction of $H^G$, and by construction of $H^G$, $ev$ is a graph morphism.\\
\indent Now let $X$ be a graph with morphism $g:X\times G\rightarrow H$. We show there is a unique morphism $\overline{g}:X\rightarrow H^G$ such that $g=ev(\overline{g}\times 1_G)$.\\
\indent Let $x\in V(X)$ and consider $\lbrace x \rbrace\times G := \lbrace (x,v)| (x,v)\in V(X\times G)$ for some $v\in V(G)\rbrace \subseteq V(X\times G).$ Then $g|_{\lbrace x \rbrace \times G}$ induces a function $f_x:V(G)\rightarrow V(H)$ defined by $f_x(v)=g((x,v))$. Then for $g=ev(\overline{g}\times 1_G)$ to hold, define $\overline{g}(x)=f_x$, and $\overline{g}$ is a vertex set function uniquely determined by $g$.\\
\indent Now let $e\in P(X)$ with $\partial_X(e)=(x_1$\textunderscore$x_2)$. Consider $\lbrace e \rbrace \times G := \lbrace d \in P(X\times G)| \partial_{X\times G}(d)=((x_1,u)$\textunderscore$(x_2,v))$ for some $u,v\in V(G) \rbrace \subseteq P(X\times G)$. Note that for a part $d\in \lbrace e \rbrace \times G$, $\partial_{X\times G}(d)=((x_1,u)$\textunderscore$(x_2,v))$ for some $u,v\in V(G)$ implies there is a part $d^\prime \in P(G)$ such that $\partial_G(d^\prime)=(u$\textunderscore$v)$.\\
\indent For such a $d$, since $g$ preserves incidence, $\partial_H(g(d))=(g(x_1,u)$\textunderscore$g(x_2,v))=(f_{x_1}(u)$\textunderscore$f_{x_2}(v))$.
Then for $g=ev(\overline{g}\times 1_G)$ to hold, define $\overline{g}(e)=a$ where $\partial_{H^G}(a)=(f_{x_1}$\textunderscore$f_{x_2})$ which exists by definition of $H^G$, and is uniquely determined by $g$. Clearly $\overline{g}$ is a morphism in $\mathbf{SiLlGraphs}$ and is uniquely determined by $g$.\\
\indent \textbf{Axiom L3 (``subobject classifier'') fails for} \textbf{SiLlGrphs}: 
Assume a subobject classifier, $\Omega$, exists with morphism $\top:\hat{1}\rightarrow\Omega$. Consider $K_2^c$ having vertices $a$ and $b$ with $!_{K_2^c}:K_2^c\rightarrow\hat{1}$ the unique morphism to the terminal object. Let $i:K_2^c\hookrightarrow K_2$ be inclusion where $K_2$ is $K_2^c$  with edge $e$. Then there exists a unique $\chi_{K_2^c}:K_2\rightarrow\Omega$ such that $K_2^c$ is the pullback of $\top$ and $\chi_{K_2^c}$. Then $\top !_{K_2^c}=\chi_{K_2^c}i$.
\begin{figure}[h]
\centering \includegraphics[scale=.6]{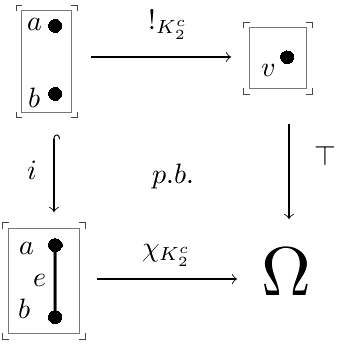}
\caption{A picture for the counterexample to the existence of a subobject classifier in \textbf{SiLlGrphs}.}
\end{figure}\par
Since $!_{K_2^c}(a)=!_{K_2^c}(b)=v$ for $v$ the vertex of $\hat{1}$ and, and since morphisms much map vertices to vertices, $\top(!_{K_2^c}(a))=\top(!_{K_2^c}(b))=\top(v)$. Since $\top !_{K_2^c}=\chi_{K_2^c} i$, $\chi_{K_2^c}(i(a))=\chi_{K_2^c}(i(b))=\top(v)$. Since graphs in $\mathbf{SiLlGraphs}$ are loopless and incidence must be preserved, $\chi_{K_2^c}(e)=\top(v)$.\\
\indent Now consider the pullback of $\chi_{K_2^c}$ and $\top$. It is the vertex induced subgraph of $K_2\times\hat{1}$ on $V(Pb)=\lbrace (c,v)\in V(K_2\times\hat{1})|\chi_{K_2^c}(\pi_{K_2}((c,v)))=\top(\pi_{\hat{1}}((c,v)))\rbrace$. However, since $K_2\times\hat{1}\cong K_2$ and $\chi_{K_2^c}(\pi_{K_2^\ell}((a,v)))=\chi_{K_2^c}(a)=\top(v)=\chi_{K_2^c}(b)=\chi_{K_2^c}(\pi_{K_2}((b,v)))$, $V(Pb)=\lbrace (a,v),(b,v)\rbrace$ and $Pb\cong K_2$. This contradicts that $K_2^c$ is the pullback of $\chi_{K_2^c}$ and $\top$. Hence no subobject classifier exists.\\
\indent \textbf{Axioms L4 (``natural number object'') holds for } \textbf{SiLlGrphs}, \textbf{L5 and L6 (``choice'' and ``two-valued'') fails for} \textbf{SiLlGrphs}: The natural number object of \textbf{SiLlGrphs} is the same as in \textbf{Grphs}, and the counterexample to choice in \textbf{Grphs} applies as well. Since \textbf{SiLlGrphs} does not have a subobject classifier, Axiom 6 does not apply.
\end{proof}

\subsection{Lawvere-type Axioms for \textbf{StGrphs}}

\begin{proposition}
\textbf{StGrphs} satisfies axioms (L1), (L3), and (L4) and does not satisfy axioms (L2), (L5), and (L6).
\end{proposition}
\begin{proof}
\textbf{Axiom L1 (``limits'') holds for} \textbf{StGrphs}: As in the proof for \textbf{Grphs}, we will show the existence of limits and colimits by constructing the terminal object, the product, the equalizer, and their duals. It is easily shown that the graph with a single vertex and a loop, $K_1^\ell$, is the terminal object and the empty vertex set (and part set) graph $\emptyset$ is the initial object.\\
\indent For products, since strict morphisms map vertices to vertices the vertex set will be precisely the same as in the \textbf{Grphs} product, $V(A\times B) := V(A)\times V(B)$. However, the edges in the product will be different from the product in \textbf{Grphs}. In order to find the part set of the product of $A$ and $B$, we  take the product in \textbf{Grphs} and delete every part which is labeled $(e,f)$ where exactly one of $e$ or $f$ was a vertex (these edges were projected to a vertex in one graph and an edge in the other which is impossible with strict morphisms). The projection morphisms are based on the first and second coordinate as they were in \textbf{Grphs}.
\begin{figure}[h]
\centering \includegraphics[scale=.7]{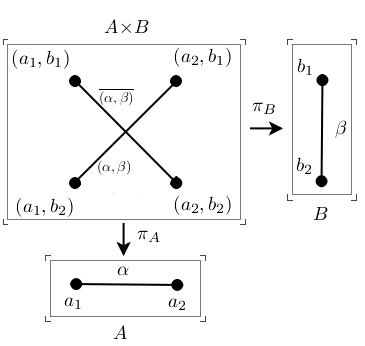}
\caption{An example of the product in \textbf{StGrphs} with pictures.}
\end{figure}\\
\indent The coproduct and equalizer in \textbf{StGrphs} is precisely the same as in \textbf{Grphs}.\\
\indent For the coequalizer in \textbf{StGrphs} we do not need to check again that identifying equivalent edges gives us a graph since the graph object is the same and the morphisms are a subcollection of the morphisms from \textbf{Grphs}. We note, however, that the equivalence classes in \textbf{StGrphs} can contain either edges or vertices but not both.\\
\indent \textbf{Axiom L2 (``exponentiation with evaluation'') fails for} \textbf{StGrphs}: Suppose \textbf{StGrphs} did have exponentiation with evaluation. Then it has a terminal object, products and exponentiation with evaluation, and therefore there is a standard adjoint functor relationship providing a bijection between the set of morphisms $X\times A \rightarrow B$ and the set of morphisms $X \rightarrow B^A$.\par
Let $B$ be a vertex with two loops, and let $A$ be two vertices with two edges between them. As in \textbf{Grphs}, we will try to determine what $B^A$ would have to be, if it existed, by testing various choices of $X$ and counting morphisms.\par
Let's begin with an $X$ as just a vertex. The product $X\times A$ is simply two vertices, and there is only one morphism from this product to $B$ (both vertices go to the only vertex in $B$). This means that there is only one morphism from a single vertex to $B^A$, therefore $B^A$ must have exactly one vertex.\par
Now let's test this object against a different $X$. Let $X$ be a vertex with one loop (note that this is the multiplicative identity, $\hat{1}$, for \textbf{StGrphs}). Hence $X\times A\cong A$.\par 
Both of the vertices in $X\times A$ must go to the only vertex in $B$, and there are two choices where to send each of the edges in $X\times A$ (since both vertices are sent to a single vertex the each edge can go to either loop independent of where the other edge was sent, and since these are strict morphisms an edge must go to an edge), so there are $2*2=4$ morphisms here. Therefore there must be exactly $4$ morphisms from $X$ to $B^A$. We've already determined that $B^A$ has exactly one vertex, and since $X$ is a vertex with a loop, $B^A$ must be a vertex with four loops.\par
Now we will test this object against one more $X$ to derive a contradiction. Let $X$ be two vertices and a single edge between them.\par
\begin{figure}[h]
\centering\includegraphics[scale=.6]{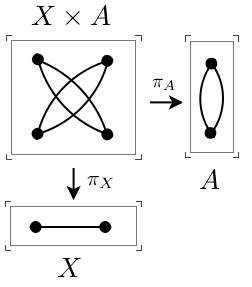}
\caption{A picture of $X\times A$, for the third test choice of $X$.}
\end{figure}
Again, since $B$ has only one vertex, all four vertices of $X\times A$ must be sent there. Which leaves us with four edges and two places where each can be sent (and each choice of where to send an edge is independent of the others), so we have $2^4=16$ morphisms here. But there are only four strict morphisms from $X$ to $B^A$. So exponentiation with evaluation does not exist in \textbf{StGrphs}.\\
\indent \textbf{Axiom L3 (``subobject classifier'') holds for} \textbf{StGrphs}: The subobject classifier in \textbf{StGrphs} is the following graph:\par
\begin{figure}[h]
\centering\includegraphics[scale=.6]{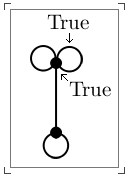}
\caption{A picture of the subobject classifier $\Omega$ in \textbf{StGrphs}}
\end{figure}
together with the canonical strict morphism $\top:\hat{1}\rightarrow \Omega$ from the terminal object $\hat{1}$, which sends the vertex and the loop of $\hat{1}$ to the vertex and loop labeled ``True" above.\par
For any subgraph $A$ of a graph $X$ we must map every vertex of $A$ to the ``True" vertex in $\Omega$ and every edge of the image of $A$ to the ``True" edge in $\Omega$. Every vertex in $V(X)\backslash V(A)$ must then be sent to the other vertex of $\Omega$. This then forces where the image of any other edge not in the image of $A$ is mapped. If the edge has both its incident vertices in the image of $A$ it is mapped to the loop at the ``true" vertex which is not labeled ``true". If the edge has one incident vertex in the image of $A$ and the other incident vertex in $V(X)\backslash V(A)$ it is mapped to the non-loop edge between the two vertices of $\Omega$. If the edge has both its incident vertices in $V(X)\backslash V(A)$ it is mapped to the loop at the vertex which is not labeled ``true". This construction satisfies the universal mapping property for subobject classifiers.\\
\indent \textbf{Axiom L4 (``natural number object'') holds for} \textbf{StGrphs}: The natural number object in \textbf{StGrphs}, $\mathbb{N}$, will be countably many vertices with loops labeled with the natural numbers, coupled with the initial morphism $\ulcorner 0 \urcorner : \hat{1} \rightarrow \mathbb{N}$ defined by mapping the single vertex and loop of the terminal object to the vertex and loop labeled $0$, and successor function $\sigma: \mathbb{N} \rightarrow \mathbb{N}$ where given a vertex with a loop labeled $n$, $\sigma(n)=n+1$. The natural number object works similarly as it did in \textbf{Grphs} only now instead of going through successive vertices, it can only go through successive vertices with loops.\\
\indent \textbf{Axiom L5 (``choice'') fails for} \textbf{StGrphs}: The same counterexample from \textbf{Grphs} applies here.\\
\indent \textbf{Axiom L6 (``two-valued'') fails for} \textbf{StGrphs}: By applying the definition of terminal object and coproduct we have that $1+1$ is a graph with two vertices and a single loop at each vertex. But $\Omega$ has a non-loop edge, so it can not be isomorphic to $\Omega$ and is therefore not two-valued.
\end{proof}
\subsection{Lawvere-type Axioms for \textbf{SiLlStGrphs}}
\begin{proposition}
\textbf{SiLlStGrphs} does not satisfy any of the axioms of \textbf{Sets}.
\end{proposition}
\begin{proof}
\textbf{Axiom L1 (``limits'') fails for} \textbf{SiLlStGrphs}: We prove that \textbf{SiLlStGrphs} fails to have limits and colimits by proving that no terminal object exists and coequalizers do not exist.\\
\indent We show no terminal object exists by examining the two cases for a graph $G$ in \textbf{SiLlStGrphs}. Either $G$ has no edges, or $G$ has an edge. If $E(G)=\emptyset$, as strict graph homomorphisms must send edges to edges, a graph that does contain an edge does not admit a strict graph homomorphism to $G$. Hence, $G$ cannot be a terminal object.\\
\indent If there is an edge $e\in E(G)$, since the graphs in \textbf{SiLlStGrphs} are loopless, $\partial_G(e)=(u$\textunderscore$v)$ for some $u,v \in V(G)$ where $u$ and $v$ are distinct. The consider the morphisms from $K_1$, the graph containing only a single vertex, $w$, to $G$. Since $u$ and $v$ are distinct, there are two distinct morphisms, $f,g:K_1\rightarrow G$ defined by $f(w)=u$ and $g(w)=v$. Hence $G$ is not a terminal object since not every graph admits a unique morphism to $G$. Hence no terminal object exists in \textbf{SiLlStGrphs}\\
\indent Assume coequalizers exist. Let $A=B=K_2$, the complete graph on 2 vertices $a$ and $b$ with edge $e$, and consider the following two morphisms $1_{K_2},t_w:A\rightarrow B$ where $1_{K_2}$ is the identity morphism and $t_w$ is the morphism where $t_w(a)=b$, $t_w(b)=a$, and $t_w(e)=e$. The coequalizer, $Coeq$ with morphism $c:B\rightarrow Coeq$ such that $c 1_{K_2}=c t_w$, exists by hypothesis. Since $c 1_{K_2}=c t_w$ and $1_{K_2}(e)=t_w(e)=e$, $c(1_{K_2}(e))=c(e)=c(t_w(e))$, and since morphisms must send edges to edges, $c(e)$ is an edge of $Coeq$.\\
\indent Let $\partial_{Coeq}(c(e))=(u$\textunderscore$v)$ for some $u,v\in V(Coeq)$. Then since morphisms preserve incidence, $c(1_{K_2}(a))=c(a)$ is incident to $c(e)$, and $c(1_{K_2}(b))=c(b)$ is incident to $c(e)$. Hence $c(a)=u$ or $c(a)=v$.\\
\indent Without loss of generality, let $c(a)=u$. Then $c(b)=v$, and since $c(1_{K_2}(a))=c(t_w(a))$, $u=c(a)=c(b)=v$. Hence $e$ is a loop of $Coeq$, which contradicts our hypothesis that $Coeq$ was in \textbf{SiLlStGrphs}. Thus coequalizers do not exist in \textbf{SiLlStGrphs}.\\
\indent \textbf{Axiom L2 (``exponentiation with evaluation'') fails for} \textbf{SiLlStGrphs}: Suppose \textbf{StLlStGrphs} did have exponentiation with evaluation. Then as it has binary products and exponentiation with evaluation, there is a standard adjoint functor relationship providing a bijection between the set of morphisms $X\times A \rightarrow B$ and the set of morphisms $X \rightarrow B^A$.\par
Let $A$ and $B$ be $K_2$, the graph with two vertices joined by a single edge. As in \textbf{Grphs}, we will try to determine what $B^A$ would have to be, if it existed, by testing various choices of $X$ and counting morphisms.\par
We begin with a test object $X=K_1$ or a graph with only a vertex. The product $X\times A$ is simply two vertices, and there are only four morphisms from this product to $B$. This means that there is four morphisms from a single vertex to $B^A$, therefore $B^A$ must have exactly four vertices.\par
Now test this object against a different $X$. Let $X$ be another copy of $K_2$. Then $X\times A$ is two disjoint copies of $K_2$. As edges must be sent to edges by strict morphisms, there are four morphisms from $X\times A$ to $B$ (two choices for each edge of $X \times A$). Hence there are four morphisms from $X$ to $B^A$. As there are no loops, each pair of morphisms from $X=K_2$ to a $B^A$ identifies as edge of $B^A$. Thus $B^A$ has two edges. \par 
Finally, we will test $X=K_3$, the complete graph of three vertices. Now $X\times A$ is an isomorphic copy of $C_6$ or the cycle graph on six vertices. There are two morphisms from $C_6$ to $B$ as a morphism of $C_6$ to $B$ is determined by how a single edge and its two incident vertices are mapped into $B$. Thus there are two morphisms from $X$ to $B^A$. However, as $X$ is isomorphic to $K_3$, this implies $B^A$ contains an odd cycle. However, $B^A$ has only two edges and cannot contain an odd cycle, a contradiction. Hence exponentiation with evaluation does not exist in \textbf{SiLlStGrphs}.\\
\indent \textbf{Axioms L3, L4, L5, and L6 (``subobject classifier'', ``natural number object'', ``choice'' and ``two-valued'') fails for} \textbf{SiLlStGrphs}: As no terminal object exists in \textbf{SiLlStGrphs}, neither does a subobject classifier, nor does a natural number object. Since no subobject classifier exists, axiom 6 does not apply. The same counterexample for choice in \textbf{Grphs} applies here.
\end{proof}
\indent We note that products, coproducts, equalizers, and an initial object exist in \textbf{SiLlStGrphs}, by the following constructions in \textbf{SiStGrphs}.\\

\subsection{Lawvere-type Axioms for \textbf{SiStGrphs}}
\indent As \textbf{SiStGrphs} is the most common category of graphs studied in literature, axioms (L1) and (L2) in the following result are already known \cites{HN2004, Dochtermann}, but we include it for completeness.
\begin{proposition}
\textbf{SiStGrphs} satisfies axioms (L1), (L2), and (L4) and does not satisfy axioms (L3), (L5), and (L6).
\end{proposition}
\begin{proof}
Axioms (L1) (``limits'') and (L2) (``exponentiation with evaluation'') hold for \textbf{SiStGrphs} with proofs in \cites{HN2004, Dochtermann}.\\
\indent \textbf{Axiom L3 (``subobject classifier'') fails for} \textbf{SiStGrphs}: Assume a subobject classifier, $\Omega$, exists with morphism $\top:\hat{1}\rightarrow\Omega$. Consider $K_2$ having vertices $a$ and $b$ with an edge $e$ between them with $!_{K_2}:K_2\rightarrow\hat{1}$ the unique morphism to the terminal object. Let $i:K_2\hookrightarrow K_2^\ell$ be inclusion where $K_2^\ell$ is $K_2$ together with a loops $\ell_a$ and $\ell_b$ at vertices $a$ and $b$ respectively. Then there exists a unique $\chi_{K_2}:K_2^\ell\rightarrow\Omega$ such that $K_2$ is the pullback of $\top$ and $\chi_{K_2}$. Then $\top !_{K_2}=\chi_{K_2}i$.
\begin{figure}[h]
\centering \includegraphics[scale=.6]{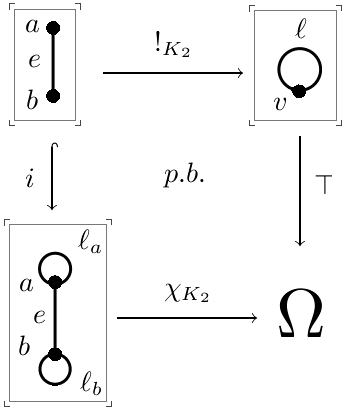}
\caption{A picture for the counterexample to the existence of a subobject classifier in \textbf{SiStGrphs}.}
\end{figure}\\
\indent Since $!_{K_2}(a)=!_{K_2}(b)=v$ for $v$ the vertex of $\hat{1}$ and $!_{K_2}(e)=\ell$ for $\ell$ the loop of $\hat{1}$, and since morphisms much send edges to edges, $\top(!_{K_2}(a))=\top(!_{K_2}(b))=\top(v)$ and $\top(!_{K_2}(e))=\top(\ell)$ where $\partial_{\Omega}(\top(\ell))=(\top(v)$\textunderscore$\top(v))$. Since $\top !_{K_2}=\chi_{K_2} i$, $\chi_{K_2}(i(a))=\chi_{K_2}(i(b))=\top(v)$. Then since morphisms preserve incidence, $\partial_{\Omega}(\chi_{K_2}(\ell_a))=\partial_{\Omega}(\chi_{K_2}(\ell_b))=(\top(v)$\textunderscore$\top(v))$. Since graphs in \textbf{SiStGrphs} can have at most one loop at any vertex, and morphisms must send edges to edges, $\chi_{K_2}(\ell_a)=\chi_{K_2}(\ell_b)=\top(\ell)$.\\
\indent Now consider the pullback of $\chi_{K_2}$ and $\top$. It is the vertex induced subgraph of $K_2^\ell\times\hat{1}$ on $V(Eq)=\lbrace (c,v)\in V(K_2^\ell\times\hat{1})|\chi_{K_2}(\pi_{K_2^\ell}((c,v)))=\top(\pi_{\hat{1}}((c,v)))\rbrace$. However, since $K_2^\ell\times\hat{1}\cong K_2^\ell$ and $\chi_{K_2}(\pi_{K_2^\ell}((a,v)))=\chi_{K_2}(a)=\top(v)=\chi_{K_2}(b)=\chi_{K_2}(\pi_{K_2^\ell}((b,v)))$, $V(Eq)=\lbrace (a,v),(b,v)\rbrace$ and $Eq\cong K_2^\ell$. This contradicts that $K_2$ is the pullback of $\chi_{K_2}$ and $\top$. Hence no subobject classifier exists.\\
\indent \textbf{Axioms L4 (``natural number object'') holds for } \textbf{SiStGrphs}, \textbf{L5 and L6 (``choice'' and ``two-valued'') fails for} \textbf{SiStGrphs}: The natural number object for \textbf{SiStGrphs} is the same as in \textbf{StGrphs}, and the counterexample in \textbf{Grphs} of choice applies here as well. Since \textbf{SiStGrphs} does not have a subobject classifier, axiom 6 does not apply.
\end{proof}
\indent Note that the three axioms (L1)-(L3) define a topos and that \textbf{SiStGrphs} and \textbf{SiLlGrphs} are missing a subobject classifier, while \textbf{SiGrphs}, \textbf{StGrphs}, and \textbf{Grphs} are missing exponentiation with evaluation. We provide a reference table for the Lawvere-type Axioms (Table 1).\\
\begin{table}[h]
\caption{Lawvere-type Axioms for categories of graphs.}
\begin{tabular}{ r || l || l | l | l | l | l | l | }			
  \quad & \textbf{Sets} & \textbf{SiLlStG} & \textbf{SiLlG} & \textbf{SiStG} & \textbf{SiG} & \textbf{StG} & \textbf{G} \\
\hline
\hline
 (L1) Limits & Y & N & Y & Y & Y & Y & Y \\
\hline  
  (Colimits) & Y & N & Y & Y & Y & Y & Y \\
\hdashline
  $\hat{1}$ & Y & N & Y & Y & Y & Y & Y \\
\hline 
  $(\hat{0})$ & Y & Y & Y & Y & Y & Y & Y \\
\hline 
  $\times$ & Y & Y & Y & Y & Y & Y & Y  \\
\hline 
  $(+)$ & Y & Y & Y & Y & Y & Y & Y \\
\hline 
  Equalizer & Y & Y & Y & Y & Y & Y & Y \\
\hline 
  (Coequalizer) & Y & N & Y & Y & Y & Y & Y \\
\hline 
\hline
 (L2) Exp. with Eval. & Y & N & Y & Y & N & N & N \\
\hline 
\hline
 (L3) Subobj. Classifier & Y & N & N & N & Y & Y & Y \\
\hline 
\hline
\hline
 (L4) Nat. Num. Obj. & Y & N & Y & Y & Y & Y & Y \\
\hline 
\hline
 (L5) Choice & Y & N & N & N & N & N & N \\
\hline 
\hline
 (L6) 2-valued & Y & N & N & N & N & N & N \\
\hline 
\end{tabular}
\end{table}
\\
\section{Other Categorial Properties of the Categories of Graphs}
\subsection{Epis and Monos in Categories of Graphs}
\indent We begin our investigation of other properties by first giving characterizations of epimorphisms and monomorphisms in the categories of graphs. These characterizations of epimorphisms and monomorphisms were known in \textbf{Grphs} \cite{KKWil}.
\begin{proposition}
A morphism in \textbf{Grphs}, \textbf{StGrphs}, and \textbf{SiGrphs} is an epimorphism if and only if it is a surjective function of the part sets, and a morphism in the above categories is a monomorphism if and only if it is an injective function of the part sets.
\end{proposition}
\begin{proof}
It is trivial to show that surjections are always epimorphisms and injections are always monomorphisms, we prove the converses.\\
\indent So let $f:A\rightarrow B$ be an epimorphism in \textbf{Grphs}, and suppose $f$ is not surjective. Then there exists $e\in P(B)\backslash f_P(P(A))$.\\
\indent First suppose $e\in V(B)$. Construct the graph $C$ by appending a vertex $e^\prime$ to $B$ such that $e^\prime$ is adjacent to every vertex $v$ is adjacent to. By construction $B$ is a subgraph of $C$.\\
\indent Since $e\in P(B)\backslash f_P(P(A))$, no edge incident to $e$ is in the image of $f$. Now consider $i:B\rightarrow C$ the inclusion morphism and $g:B\rightarrow C$ defined by $g(u)=i(u)$ for all $u\in V(B)\backslash \lbrace e \rbrace$, $g(e)=e^\prime$, $g(m)=i(m)$ for all edges $m$ not incident to $e$, and for edge $n$ incident to $e$, set $g(n)$ to be the corresponding edge incident to $e^\prime$. This is clearly a morphism (actually it is strict). Then $if=gf$ but $i\neq g$, a contradiction to $f$ being an epimorphism.\\
\indent Now suppose $e$ is an edge of $B$. Construct the graph $C$ by appending an edge $e^\prime$ to $B$ such that $e^\prime$ has the same incidence as $e$. Then by construction $B$ is a subgraph of $C$.\\
\indent Now consider $i:B\rightarrow C$ the inclusion morphism and $g:B\rightarrow C$ defined by $g(u)=i(u)$ for all $u\in P(B)\backslash\{e\}$ and $g(e)=e^\prime$. As the incidence of $e^\prime$ is the same as $e$ this is a morphism (it is actually strict). Then $if=gf$ but $i\neq g$, a contradiction to $f$ being an epimorphism. Hence epimorphisms in \textbf{Grphs} are surjective functions of the corresponding part sets. A similar proof applies or \textbf{StGrphs}.\\
\indent For \textbf{SiGrphs}, a similar proof applies. However, in the case that $e\in P(B)\backslash f_P(P(A))$ is an edge, a different construction is required. Let $C$ be $K_1^\ell$, the graph with one vertex and one loop. Let $h:B\rightarrow C$ be the morphism that maps everything to the vertex, and $g:B\rightarrow C$ be the morphism that maps everything except $e$ to the vertex, and maps $e$ to the edge. Then $hf=gh$ but $h \neq g$, a contradiction.\\
\indent Now let $f:A\rightarrow B$ be an monomorphism in \textbf{Grphs}, and suppose $f$ is not injective. Then there exists $d,e\in P(A)$ such that $f(d)=f(e)$. Consider $g,h:K_2\rightarrow A$ where $g$ maps the edge to $d$, and the vertices to the vertices incident to $d$ and $h$ maps the edge to $e$ and the vertices to the vertices incident to $e$. Then as $f$ must preserve incidence, $fg=fh$ but $g\neq h$, a contradiction to $f$ being a monomorphism. A similar proof applies to \textbf{StGrphs} and \textbf{SiGrphs}
\end{proof}
\indent This changes if we add enough restrictions, as seen in the following proposition.
\begin{proposition}
A morphism in \textbf{SiStGrphs}, \textbf{SiLlGrphs}, and \textbf{SiLlStGrphs} is an epimorphism if and only if it is a surjective function of vertex sets, and a morphism in the above categories is a monomorphism if and only if it is an injective function of the vertex sets.
\end{proposition}
\begin{proof}
Let $f:A\rightarrow B$ be an epimorphism in \textbf{SiLlGrphs}. Suppose $f_V$ is not surjective. Then there exists $v\in V(B)\backslash f_V(V(A))$. Construct the graph $C$ by appending a vertex $v^\prime$ to $B$ such that $v^\prime$ is adjacent to every vertex $v$ is adjacent to. By construction $B$ is a subgraph of $C$.\\
\indent Since $v\in V(B)\backslash  f_V(V(A))$, no edge incident to $v$ is in the image of $f_P$. Now consider $i:B\rightarrow C$ the inclusion morphism and $g:B\rightarrow C$ defined by $g(u)=i(u)$ for all $u\in V(B)\backslash \lbrace v \rbrace$, $g(v)=v^\prime$, $g(e)=i(e)$ for all edges $e$ not incident to $v$, and for edge $f$ incident to $v$, set $g(f)$ to be the corresponding edge incident to $v^\prime$. Then $if=gf$ but $i\neq g$, a contradiction to $f$ being an epimorphism. Hence epimorphisms in \textbf{SiLlGrphs} have surjective vertex set functions. A similar proof applies to \textbf{SiStGrphs} and \textbf{SiLlStGrphs}.\\
\indent Suppose $f:A\rightarrow B$ is a morphism in \textbf{SiLlGrphs} and $f_V$ is surjective. Consider morphisms $h,k:B\rightarrow C$ such that $hf=kf$. Since $f_V$ is surjective and $h_Vf_V=k_Vf_V$, $h_V=k_V$. So if $h\neq k$ there exists an edge $e\in E(B)$ such that $h(e)\neq k(e)$, even though $h_V=k_V$. There are two possibilities for $h(e)$ and $k(e)$, either as different vertices or edges.\\
\indent If $h(e)$ and $k(e)$ are different vertices, as $h_V=k_V$, the incident vertices to $e$ in $B$ are both mapped to the same vertex, so for incidence to hold $h(e)$ and $k(e)$ would also be mapped to that vertex and $h(e)=k(e)$. If $h(e)$ and $k(e)$ are mapped to different edges, since $h_V=k_V$ they must have the same incidence. Since graphs in \textbf{SiLlGrphs} are simple and loopless, $h(e)=k(e)$. Hence both possibilities lead to contradictions.
A similar proof holds for \textbf{SiLlStGrphs}, and for \textbf{SiStGrphs} a third possiblity arises for $h(e)$ and $k(e)$ to be different loops. However, in this case, as simple graphs have only one loop and $h_V=k_V$, they must be mapped to the same loop.\\
\indent Now let $f:A\rightarrow B$ be a monomorphism in \textbf{SiLlGrphs}. Suppose $f_V$ is not injective. Then there exists $u,v\in V(A)$ such that $f(u)=f(v)$. Then consider the two morphisms $j,k:K_1\rightarrow A$ defined by $j$ mapping the single vertex of $K_1$ to $u$, and $k$ mapping the single vertex of $K_1$ to $v$. Clearly $fj=fk$ but $j\neq k$ a contradiction to $f$ being a monomorphism. Hence monomorphisms in \textbf{SiLlGrphs} have injective vertex set functions. A similar proof applies to \textbf{SiStGrphs} and \textbf{SiLlStGrphs}.\\
\indent Suppose $f:A\rightarrow B$ is a morphism in \textbf{SiLlGrphs} and $f_V$ is injective. Consider morphisms $j,k:C\rightarrow A$ such that $fj=fk$. Since $f_V$ is injective, $j_V=k_V$. Thus if there exists $e\in P(C)$ such that $j(e)\neq k(e)$, then $e$ must be an edge of $C$. Since $f_V$ is injective and $fj=fk$, $j(e)$ and $k(e)$ cannot both be vertices in $A$. Without loss of generality assume $k(e)$ is an edge.\\
\indent Note that $j(e)$ cannot be a vertex of $A$, for both incident vertices of $e$ in $C$ are mapped to $j(e)$ as well. Then since $j_V=k_V$, morphisms preserve incidence, and the graphs are loopless, $k(e)$ is mapped to a vertex. Hence $j(e)$ must be an edge of $A$. Since $j_V=k_V$, $j(e)$ and $k(e)$ have the same incident vertices, and since the graphs are simple, $k(e)=j(e)$, a contradiction. Hence $f$ is a monomorphism.\\
\indent Similar proofs apply in \textbf{SiStGrphs} and \textbf{SiLlStGrphs}. 
\end{proof}

\subsection{Free Objects and Cofree Objects in the Categories of Graphs}
\indent We now consider the underlying vertex set functor $|-|_V:\mathbf{-Grphs}\sim\!\rightarrow\mathbf{Sets}$ defined for each of the categories of graphs, where  $|G|_V=V(G)$ for a graph $G$ and $|f|_V=f_V$ for a morphism $f$. We define the free graph functor $F(-):\mathbf{Sets}\sim\!\rightarrow\mathbf{-Grphs}$ to be such that $F(-)$ is left adjoint to $|-|_V$. We similarly define the cofree graph functor $C(-):\mathbf{Sets}\sim\!\rightarrow\mathbf{-Grphs}$ to be such that $|-|_V$ is left adjoint to $C(-)$. The following two proposition characterizes the free and cofree objects in the categories of graphs.
\begin{proposition}
Given a set $X$, in all six categories of graphs, the free graph on $X$ is just the empty edge graph on the vertex set $X$, and the free objects are the empty edge set graphs, denoted $K_n^c$ for finite vertex sets.
\end{proposition}
\begin{proof}
Let $X$ be a set in \textbf{Sets} with $n$ elements, and let $F(X)=K_n^c$ where $V(F(X))=X$. Now let $G$ be a graph such that there is a function $g:X\rightarrow |G|_V$. We show there is a unique graph morphism $\overline{g}:F(X)\rightarrow G$ such that $g=|\overline{g}|_Vu$ for some $u:X\rightarrow |F(X)|_V$. Note that $|F(X)|_V=V(F(X))=X$. Hence define the function $u:X\rightarrow |F(X)|_V$ as $u=1_X$.\\
\indent Let $\overline{g}$ be the pair of function maps $\overline{g}_V=g$ and $\overline{g}_E=g$. Since there are no edges in $F(X)$, incidence is clearly preserved (and the morphism is strict). Then since $g=|\overline{g}|_Vu$ must hold, $u=1_X$, and $|\overline{g}|_V=\overline{g}_V=g$, $\overline{g}$ is uniquely determined by $g$. 
\end{proof}
\begin{proposition}
Given a set $X$,
\begin{enumerate}
\item in \textbf{Grphs},  \textbf{SiGrphs}, and  \textbf{SiLlGrphs} the cofree graph on $X$ is the complete graph on the vertex set $X$ and the cofree objects are the complete graphs with at least one vertex, denoted $K_n$ for finite vertex sets with $n\geq 1$,
\item in \textbf{StGrphs}, \textbf{SiStGrphs} the cofree graph on $X$ is a complete graph with a single loop on each vertex having the vertex set $X$, and the cofree objects are the complete graphs with a loop at each vertex with at least one vertex, denoted $K_n^\ell$ for finite vertex sets with $n\geq 1$,
\item in \textbf{SiLlStGrphs} no cofree graph exists. 
\end{enumerate}
\end{proposition}
\begin{proof}
\textbf{Part 1}: Let $X$ be a set in \textbf{Sets} and define $C(X)$ as the complete graph with the vertex set $V(C(X))=X$. Let $G$ be a graph in \textbf{Grphs} with set function $g:|G|_V\rightarrow X$. We show that there is a unique graph morphism $\overline{g}:G\rightarrow C(X)$ such that $g=c |\overline{g}|_V$ for some set function $c:|C(X)|_V\rightarrow X$  Note that $|C(X)|_V=V(C(X))=X$. Hence we define $c$ as $1_X$.\\
\indent For $g=1_X|\overline{g}|_V$ to hold, $\overline{g}_V=g$ is uniquely determined. Then let $e$ be an edge of $G$ incident to vertices $x,y\in V(G)$ where $x$ and $y$ are not necessarily distinct. Then since graph morphisms must preserve incidence, for $\overline{g}$ to be a morphism, $\overline{g}(e)$ must map to the part $e^\prime$ of $C(X)$ incident to vertices $g(x)$ and $g(y)$. By the definition of $C(X)$ such a part $e^\prime$ exists, even if it is a vertex. Hence $\overline{g}$ exists and is uniquely determined by $g$. A similar proof applies for \textbf{SiGrphs} and \textbf{SiLlGrphs}.\\
\indent \textbf{Part 2}: Let $X$ be a set in \textbf{Sets} and define $C(X)$ as the complete graph with a loop at every vertex with the vertex set $V(C(X))=X$. Let $G$ be a graph in \textbf{StGrphs} with set function $g:|G|_V\rightarrow X$. We show that there is a unique strict graph homomorphism $\overline{g}:G\rightarrow C(X)$ such that $g=c|\overline{g}|_V$ for some set function $c:|C(X)|_V\rightarrow X$  Note that $|C(X)|_V=V(C(X))=X$. Hence we define $c$ as $1_X$.\\
\indent For $g=1_X|\overline{g}|_V$ to hold, $\overline{g}_V=g$ is uniquely determined. Then let $e$ be an edge of $G$ incident to vertices $x,y\in V(G)$ where $x$ and $y$ are not necessarily distinct. Then since strict graph homomorphisms must send edges to edges and preserve incidence, for $\overline{g}$ to be a strict graph homomorphism, $\overline{g}(e)$ must map to the edge $e^\prime$ of $C(X)$ incident to vertices $g(x)$ and $g(y)$. By the definition of $C(X)$ such an edge $e^\prime$ exists. Hence $\overline{g}$ exists and is uniquely determined by $g$. A similar proof applies for \textbf{SiStGrphs}. \\
\indent \textbf{Part 3}: Assume cofree graphs exist. Let $X=\lbrace x \rbrace$ in \textbf{Sets} and $C(X)$ be the cofree graph associated with $X$ and function $c:|C(X)|_V\rightarrow X$. Consider $K_2$ with vertices $a$ and $b$ and edge $e$ and set function $g:|K_2|_V\rightarrow X$ defined by $g(a)=g(b)=x$. Then since $C(X)$ is a cofree object, there is a unique morphism in \textbf{SiLlStGrphs}, $\overline{g}:K_2\rightarrow C(X)$, such that $g=c|\overline{g}|_V$. Since $\overline{g}$ is a strict graph homomorphism, it must send $e$ to an edge in $C(X)$. Thus $\overline{g}(e)=f$ for some $f\in E(C(X))$. Since graph homomorphisms preserve incidence, $f$ is incident to $\overline{g}(a)=|\overline{g}|_V(a)=a^\prime$ for some vertex $a^\prime\in V(C(X))$ and $\overline{g}(b)=|\overline{g}|_V(b)=b^\prime$ for some vertex $b^\prime \in V(C(X))$.\\
\indent Since $C(X)$ is loopless, $a^\prime \neq b^\prime$. Then since $g=c|\overline{g}|_V$, $g(a)=c(|\overline{g}|_V(a))=x$ and $g(b)=c(|\overline{g}|_V(b))=x$, $c(a^\prime)=c(b^\prime)=x$. Now consider the morphism $\overline{h}:K_2\rightarrow C(X)$ defined by $\overline{h}(e)=f$, $\overline{h}(a)=b^\prime$ and $\overline{h}(b)=a^\prime$. Clearly $\overline{h}\neq \overline{g}$. Then $c(|\overline{h}|_V(a))=c(b^\prime)=x=g(a)$ and $c(|\overline{h}|_V(b))=c(a^\prime)=x=g(b)$. Thus $c|\overline{h}|_V=g$, and $\overline{g}$ is not unique, which is a contradiction to the universal mapping property of the cofree object.
\end{proof}

\subsection{Projective and Injective Objects in the Categories of Graphs}
\indent The definitions for free objects and cofree objects are dependent on the category being a concrete category. We move on to other categorial constructions that are defined for any abstract category. We start with the injective objects and projective objects.
\begin{proposition}
\begin{enumerate}
\item In \textbf{Grphs}, \textbf{SiGrphs}, and \textbf{StGrphs}, all graphs with at most one edge per component are precisely the projective objects, and there are enough projective objects.
\item In \textbf{SiLlGrphs}, \textbf{SiStGrphs}, and \textbf{SiLlStGrphs}, the projective objects are precisely the free objects, and there are a enough projective objects.
\end{enumerate}
\end{proposition}
The projective objects in \textbf{Grphs} are found in \cite{KKWil}, we provide an alternate proof.
\begin{proof}
\textbf{Part 1}: First note that if $f:A\rightarrow B$ is an epimorphism in \textbf{Grphs} then $f$ is a surjective map of the associated part sets. So let $A$ be a graph with at most one edge in each component with morphism $h:A \rightarrow G$ for some graph $G$. Let $H$ be a graph with an epimorphism $g:H\rightarrow G$.\\
\indent Consider a component of $A$. If the component is composed of a single vertex, $v$ without a loop, then since $g_P$ is a surjection, there exists $v^\prime \in P(H)$ such that $h(v^\prime)=g(v)$. If the component is composed of a single vertex $v$ with a loop $\ell$, and under $g$ the loop is identified with $g(v)$, then as before there exists $v^\prime \in P(H)$ such that $h(v^\prime)=g(v)=g(\ell)$. If the loop is not identified with the vertex, then as $h$ is a surjection, there exists a loop $\ell^\prime \in P(H)$ such that $g(\ell)=h(\ell^\prime)$.\\
\indent If the component has an edge $e$, and two vertices $u$ and $v$, and under $g$ the two vertices are identified with the edge, then there exists $v^\prime \in P(H)$ such that $h(v^\prime)=g(v)=g(u)=g(e)$. If under $g$ the two vertices are identified, and the edge is sent to a loop, then there exists $v^\prime, \ell^\prime \in P(H)$ such that $h(v^\prime)=g(v)=g(u)$ and $h(\ell^\prime)=g(e)$. If under $g$ the two vertices are not identified, then $e$ is mapped to the incident edge of their images. Then since $h$ is a surjection, there exists $u^\prime, v^\prime, e^\prime \in P(H)$ such that $h(v^\prime)=g(v)$, $h(u^\prime)=g(u)$, and $h(u^\prime)=g(u)$. Then the definition for $\overline{g}$ such that $h\overline{g}=g$ is obvious, and since each component can be mapped independently from other components, this is a graph morphism.\\
\indent Now suppose $G$ is a graph with at least two edges in some component, called $e$ and $f$. Consider the graph $H$ created by ``splitting'' $G$ at each vertex incident to more than two edges. That is, for every vertex $v$ incident to at least two edges $a$ and $b$, create $v_1$ and $v_2$ in $H$ such that $a$ is incident to $v_1$ and $b$ is incident to $v_2$ with no edge between $v_1$ and $v_2$, with $v_1$ and $v_2$ replacing $v$. Then $H$ admits an epimorphism $h$ to $G$ by re-identifying these split vertices.\\
\indent However, with morphism $1_G:G\rightarrow G$, $G$ does not admit a morphism $\overline{g}$ to $H$ such that $h\overline{g}=1_G$ as edges $e$ and $f$ must be sent to the same component to preserve incidence. Hence $G$ is not projective.\\
\indent Let $G$ be a graph in \textbf{Grphs}. To show there are enough projectives, we show there is a projective object $H$ and an epimorphism $e:H\rightarrow G$. As above, construct $H$ by ``splitting'' $G$. Then $H$ admits an epimorphism to $G$, and since $H$ does not have more than one edge per component, $H$ is projective. A similar proof applies to \textbf{SiGrphs} and \textbf{StGrphs}\\
\indent \textbf{Part 2}: First note that if $f:A\rightarrow B$ is an epimorphism in \textbf{SiLlGrphs} then the vertex set function $f_V$ is surjective. We show that the free objects are projective objects. Clearly the empty graph $\emptyset$ is projective since it is the initial object. Now let $X$ be a non-empty set in \textbf{Sets}, $G$ be a graph with a morphism $h:F(X)\rightarrow G$, and $H$ be a graph with an epimorphism $g:H\rightarrow G$. We show that there is a morphism $\overline{h}:F(X)\rightarrow H$ such that $g \overline{h}=h$.\\
\indent Since $g$ is an epimorphism, $g_V$ is a surjective function. Hence for all $v_i\in V(F(X))$, there is a $u_i\in V(H)$ such that $g(u_i)=h(v_i)$. Then define $\overline{h}(v_i)=u_i$ for every $v_i\in V(F(X))$. Then $g(\overline{h}(v_i))=g(u_i)=h(v_i)$ for every vertex $v_i$ of $F(X)$. Since $F(X)$ contains no edges, $\overline{h}$ is a graph morphism (and strict). Thus $F(X)$ is projective.\\
\indent Now let $A$ be a graph with at least 1 edge, and consider $K$, the complete graph on $V(A)$, with an inclusion morphism $h:A\rightarrow K$. By Proposition 4.2. there is an epimorphism $e:K^c\rightarrow K$ for $K^c$ the empty edge graph on $V(A)$. Since $A$ has an edge any morphism from $A$ to $K^c$ must identify at least two vertices, and hence no such morphism $f:A\rightarrow K^c$ exists such that $h=ef$. Thus $A$ is not projective.\\
\indent Let $G$ be a graph in \textbf{SiLlGrphs}. To show there are enough projectives, we show there is a projective object $H$ and an epimorphism $e:H\rightarrow G$. By Proposition 4.2., the projective object $K^c$,the empty edge graph on $V(G)$, admits an epimorphism to $G$.  A similar proof applies to \textbf{SiStGrphs} and \textbf{SiLlStGrphs}.\\
\end{proof}
\begin{proposition}
\begin{enumerate}
\item In \textbf{Grphs}, \textbf{SiGrphs}, and \textbf{SiLlGrphs} the injective objects are precisely the graphs containing the cofree objects as spanning subgraphs and there are enough injective objects.
\item In \textbf{StGrphs} and \textbf{SiStGrphs}, the injective objects are precisely the graphs containing the cofree objects as spanning subgraphs and there are enough injective objects.
\item In \textbf{SiLlStGrphs}, there are no injective objects.
\end{enumerate}
\end{proposition}
The injective objects in \textbf{Grphs} are found in \cite{KKWil}, we provide an alternate proof.
\begin{proof}
\textbf{Part 1}: Let $A$ be a graph that contains a cofree spanning subgraph in \textbf{Grphs}, and let $G,H$ be graphs in \textbf{Grphs} with a morphism $f:G\rightarrow A$ and a monomorphism $g:G\rightarrow H$. We show there is a morphism $\overline{f}:H\rightarrow A$ such that $f=\overline{f}g$.\\
\indent Since $g$ is a monomorphism, it is an injection of the part sets. Then for all $v\in g_P(P(G))$ there is a unique $v^\prime \in P(G)$ such that $g(v^\prime)=v$. Since $A$ is non-empty, it has a vertex $x$. Define $\overline{f}:H\rightarrow A$ by $\overline{f}(v)=f(v^\prime)$ if $v\in g_P(P(G))$ and $\overline{f}(v)=x$ if $v$ is not in the image but a vertex. If $v$ is an edge that is not in the image with $\partial_G(v)=(u_1$\textunderscore$u_2)$, then define $\overline{f}(v)=e$ where $e$ is some edge with $\partial_A(e)=(\overline{f}(u_1)$\textunderscore$\overline{f}(u_2))$. One exists since $A$ contains a spanning cofree subgraph, and in the case that $u_1=u_2$ the vertex suffices. By this construction, $\overline{f}$ is a  morphism and $\overline{f}g=f$.\\
\indent Now let $G$ be a graph in \textbf{Grphs} that does not contain a cofree spanning subgraph. Assume it is an injective object of \textbf{Grphs}. Then there are distinct vertices $u,v\in V(G)$ such that there is no edge $e$ with $\partial_G(e)=(u$\textunderscore$v)$.\\
\indent Then consider $K_2^c$ with morphism $f:K_2^c\rightarrow G$ defined by $f(a)=u$ and $f(b)=v$, for $a$ and $b$ the two vertices of $K_2^c$, and $i:K_2^c\rightarrow K_2$ the inclusion morphism. Since the inclusion morphism is a monomorphism, there is a morphism $\overline{f}:K_2\rightarrow G$ such that $\overline{f} i = f$. Then $\overline{f}(i(a))=\overline{f}(a)=u$ and $\overline{f}(i(b))=\overline{f}(b)=v$. Since morphisms preserve incidence, $\partial_G(\overline{f}(e))=(\overline{f}(a)$\textunderscore$\overline{f}(b))=(u$\textunderscore$v)$, and there is an edge $e^\prime$ such that $\partial_G(e^\prime)=(u$\textunderscore$v)$, a contradiction. Hence $G$ is not an injective object.\\
\indent To show there are enough injective objects we show that for any graph $G$ in \textbf{Grphs}, there is an injective object $H$ with a monomorphism $f:G\rightarrow H$. If $G$ is not the initial object, $C(V(G))$ is an injective object and $i:G\rightarrow C(V(G))$, the inclusion morphism, is a monomorphism. If $G=\emptyset$ then $\emptyset\hookrightarrow K_1$ suffices. Hence there are enough injective objects in \textbf{Grphs}. A similar proof applies to \textbf{SiGrphs} as well as \textbf{SiLlGrphs} that relies on monomorphisms as injections of the vertex sets and the fact that there is at most one edge between any two distinct vertices.\\
\indent \textbf{Part 2}: Let $A$ be a graph that contains a cofree spanning subgraph in \textbf{StGrphs}, and let $G,H$ be graphs in \textbf{StGrphs} with a morphism $f:G\rightarrow A$ and a monomorphism $g:G\rightarrow H$. We show there is a morphism $\overline{f}:H\rightarrow A$ such that $f=\overline{f}g$.\\
\indent Since $g$ is a monomorphism, it is an injection of the part sets. Then for all $v\in g_P(P(G))$ there is a unique $v^\prime \in G$ such that $g(v^\prime)=v$. Since $A$ is non-empty, it has a vertex $x$. Defined $\overline{f}:H\rightarrow A$ by $\overline{f}(v)=f(v^\prime)$ if $v\in g_P(P(G))$ and $\overline{f}(v)=x$ if $v$ is not in the image but a vertex. If $v$ is an edge that is not in the image with $\partial_G(v)=(u_1$\textunderscore$u_2)$, then define $\overline{f}(v)=e$ where $e$ is some edge with $\partial_A(e)=(\overline{f}(u_1)$\textunderscore$\overline{f}(u_2))$. One exists since $A$ contains a spanning cofree subgraph. By this construction, $\overline{f}$ is a strict graph morphism and $\overline{f}g=f$.\\
\indent Now let $G$ be a graph in \textbf{StGrphs} that does not contain a cofree spanning subgraph. Assume it is an injective object of \textbf{StGrphs}. Then there are vertices $u,v\in V(G)$ (not necessarily distinct) such that there is no edge $e\in E(G)$ with $\partial_G(e)=(u$\textunderscore$v)$.\\
\indent Then consider $K_2^c$ with morphism $f:K_2^c\rightarrow G$ defined by $f(a)=u$ and $f(b)=v$, for $a$ and $b$ the two vertices of $K_2^c$, and $i:K_2^c\rightarrow K_2$ the inclusion morphism. Since the inclusion morphism is a monomorphism, there is a morphism $\overline{f}:K_2\rightarrow G$ such that $\overline{f} i = f$. Then $\overline{f}(i(a))=\overline{f}(a)=u$ and $\overline{f}(i(b))=\overline{f}(b)=v$. Since morphisms preserve incidence, $\partial_G(\overline{f}(e))=(\overline{f}(a)$\textunderscore$\overline{f}(b))=(u$\textunderscore$v)$, there is an edge $e^\prime$ such that $\partial_G(e^\prime)=(u$\textunderscore$v)$, a contradiction. Hence $G$ is not an injective object.\\
\indent To show there are enough injective objects we show that for any graph $G$ in \textbf{StGrphs}, there is an injective object $H$ with a monomorphism $f:G\rightarrow H$. If $G$ is not the initial object, $C(V(G))$ is an injective object and $i:G\rightarrow C(V(G))$, the inclusion morphism, is a monomorphism. If $G=\emptyset$ then $\emptyset\hookrightarrow K_1^\ell$ suffices. Hence there are enough injective objects in \textbf{StGrphs}. A similar proof applies to \textbf{SiStGrphs} that relies on monomorphisms as injections of the vertex set and the fact that there is at most one edge between any two (not necessarily distinct) vertices.\\
\indent \textbf{Part 3}: Suppose there exist injective objects. Let $Q$ be an injective object in \textbf{SiLlStGrphs}. Consider the complete graph $K$ with the cardinality of $V(K)$ is greater than that of $V(Q)$. Then consider the morphisms $1_Q:Q\rightarrow Q$ the identity on $Q$ and $f:Q\rightarrow K$, the inclusion morphism of $Q$ into $K$. Since inclusion morphisms are injections, they are monomorphisms.\\
\indent Then since $Q$ is injective, there is a morphism $\overline{f}:K\rightarrow Q$ such that $\overline{f} f=1_Q$. Since the cardinality of $V(K)$ is greater than the cardinality of $V(Q)$ and $\overline{f}_V$ is a set map, there are two distinct vertices $u,v\in V(K)$ such that $\overline{f}(u)=\overline{f}(v)$. Since $\overline{f}$ is a strict morphism and $K$ is a complete graph, the edge $e$ incident to $u$ and $v$ in $K$ must be sent to an edge in $Q$. Since graph homomorphisms preserve incidence and $\overline{f}(u)=\overline{f}(v)$, $\partial_Q(\overline{f}(e))=(\overline{f}(u)$\textunderscore$\overline{f}(u))$, and hence $\overline{f}(e)$ is a loop. This contradicts $Q$ being loopless. Hence no injective objects exist.
\end{proof}

\subsection{Generators and Cogenerators in the Categories of Graphs}
\indent The last property we characterize in the six categories of graphs is a classification of generators and cogenerators.
\begin{proposition}
\begin{enumerate}
\item In \textbf{Grphs} and \textbf{SiGrphs}, all graphs containing a non-loop edge are precisely the generators,
\item in \textbf{SiLlGrphs} all nonempty graphs are generators,
\item in \textbf{StGrphs} no generators exist,
\item in \textbf{SiStGrphs} and \textbf{SiLlStGrphs}, the empty edge graphs with at least one vertex are precisely the generators.
\end{enumerate}
\end{proposition}
\begin{proof}
\textbf{Part 1}: Let $A$ be a graph in \textbf{Grphs} with a non-loop edge, $e$, with vertices $u_1$ and $u_2$ incident to $e$. Consider $K_2$ with vertices $v_1$ and $v_2$ with incident edge $e^\prime$. Then $A$ has an epimorphism $f:A\rightarrow K_2$ defined by $f(u_2)=v_2$ and $f(y)=v_1$ for all vertices $y\in V(A)\backslash \{u_2\}$ and where every loop incident to $u_2$ is mapped to $v_2$ and every non-loop edge incident to $v_2$ (including $e$) is mapped to $e^\prime$, and all other edges mapped to $v_1$.\\
\indent Hence, we only need to show $K_2$ is a generator. Let $f,g:G\rightarrow H$ be such that $f\neq g$. Hence $f(a)\neq g(a)$ for some $a \in P(G)$. First suppose $a$ is a vertex. Then the morphism $\ulcorner a \urcorner$ from $K_2$ to $G$ mapping the two vertices and edge of $K_2$ to $a$ suffices. If $a$ is an edge of $G$, then the morphism that maps the edge of $K_2$ to $a$ and the incident vertices of the edge to the incident vertices of $a$ suffices.\\
\indent Now suppose $A$ is a graph containing no non-loop edges. Then no morphism from $A$ to $K_2$ can distinguish between $f,g:K_2\rightarrow K_1^\ell$, where $f$ maps the two vertices and edge of $K_2$ to the vertex of $K_1^\ell$ and $g$ maps the two vertices of $K_2$ to the single vertex of $K_1^\ell$ and the edge to the loop. Hence $A$ is not a generator. The same proof applies to \textbf{SiGrphs}.\\
\indent \textbf{Part 2}: The empty graph is in the initial object of \textbf{SiLlGrphs} and thus cannot be a generator.\\
\indent Since in all graphs of \textbf{SiLlGrphs} there is at most one edge between any two distinct vertices, if $f,g:G\rightarrow H$ agree on the vertex sets, they agree on the edge sets and $f=g$. Hence if $f,g:G\rightarrow H$ are distinct, then for some vertex $v$ of $G$, $f(v)\neq g(v)$. So let $A$ be a non-empty graph. Then $A$ has a vertex and the morphism $h:A\rightarrow G$, where $h$ maps all of $P(A)$ to the vertex $v$ suffices.\\
\indent \textbf{Part 3}: Suppose a generator $G$ in \textbf{StGrphs} exists. Then consider $f,g:K_1\rightarrow K_2^c$, where $f$ maps the vertex of $K_1$ to one vertex of $K_2^c$ and $g$ maps the vertex of $K_1$ to the other vertex of $K_2^c$. Since $G$ is a generator, it admits a morphism $h:G\rightarrow K_1$ such that $fh\neq gh$. Since morphisms are strict, edges must be mapped to edges. However, $K_1$ has no edge, and thus $G$ is edgeless.\\
\indent Now consider $K_2$ and  $A$, where $A$ is a graph consisting of two vertices with two parallel edges between the two vertices. Define $j,k:K_2 \rightarrow A$ by $j$ mapping the edge of $K_2$ to one edge of $A$, and $k$ mapping the edge of $K_2$ to the other edge of $A$, but mapping the vertices of $K_2$ in tandum. Then no morphism from $G$ can distinguish between $j$ and $k$ as $G$ has no edges. Hence $G$ is not a generator, a contradiction, and no generators exist in \textbf{StGrphs}.\\
\indent \textbf{Part 4}: First we show $K^c$, an empty edge graph with at least one vertex, is a generator, then we show that any graph with an edge is not a generator. Let $X$ and $Y$ be graphs in \textbf{SiStGrphs} with morphisms $f,g:X\rightarrow Y$ such that $f\neq g$. Then there is a vertex $v\in V(X)$ such that $f(x)\neq g(x)$, otherwise since the morphisms preserve incidence and there is at most one edge between any two vertices, $f(e)=g(e)$ for all edges $e\in P(X)$ and $f=g$.\\
\indent First note $K_1=K_1^c$. Now consider the map $h:K_1\rightarrow X$ that sends the single vertex of $K_1$, $u$, to $v$. Then $f(h(u))=f(v)\neq g(v)=g(h(u))$. Hence $f h \neq g h$. Hence $K_1$ is a generator.\\
\indent To show $K^c$ is a generator, we consider the morphism $\pi:K^c\rightarrow K_1$ that sends every vertex of $K^c$ to $u$. Then clearly $f (h \pi)\neq g (h \pi)$, and hence $K^c$ is a generator of \textbf{SiStGrphs}.\\
\indent Now let $G$ be a graph in \textbf{SiStGrphs} with at least one edge. Consider $K_2^c$ with two vertices $u$ and $v$ with two morphisms $1_{K_2^c},t_w:K_2^c\rightarrow K_2^c$ where $1_{K_2^c}$ is the identity morphism and $t_w$ is the ``twist'' morphism defined by $t_w(u)=v$ and $t_w(v)=u$. Clearly $1_{K_2^c} \neq t_w$, but since morphisms must send edges to edges, $G$ admits no map to $K_2^c$. Thus $G$ is not a generator. A similar proof applies for \textbf{SiLlStGrphs}.\\
\end{proof}
\begin{proposition}
\begin{enumerate}
\item In \textbf{Grphs} and \textbf{SiGrphs}, the graphs containing a loop and a non-loop edge are precisely the cogenerators,
\item In \textbf{SiLlGrphs}, the graphs containing an edge are the cogenerators,
\item in \textbf{StGrphs}, the graphs containing both a vertex with two distinct loops and containing a subgraph isomorphic to $K_2^\ell$ are precisely the cogenerators,
\item in \textbf{SiStGrphs}, the cogenerators are precisely the graphs containing a subgraph isomorphic to $K_2^\ell$,
\item in \textbf{SiLlStGrphs} no cogenerators exist.
\end{enumerate}
\end{proposition}
\begin{proof}
\textbf{Part 1}: Let $A$ be a graph composed of two disconnected components. One component contains a vertex $v$ with a loop $\ell$, and the other component contains two vertices $u_1$ and $u_2$ with an edge $e$ between them. If $A$ is a cogenerator, then any graph containing a loop and a non-loop edge is also a cogenerator by the appropriate inclusion morphism.\\
\indent Let $f,g:G\rightarrow H$ be two distinct morphisms in \textbf{Grphs}. Then $f(e^\prime)\neq g(e^\prime)$ for some $e^\prime\in P(G)$. If both $f(e^\prime)$ and $g(e^\prime)$ are vertices, define $h:H\rightarrow A$ by $h(f(e^\prime))=u_2$, $h(y)=u_1$ for all $y\in V(H)\backslash\{ f(e^\prime)\}$, $h(a)=e$ for all non-loop edges $a$ incident to $h(e^\prime)$ in $H$, $h(a)=u_2$ for all loops $a$ incident to $h(e^\prime)$, and $h(a)=u_1$ for all other edges $a$ of $H$. Hence $hf\neq hg$.\\
\indent If $f(e^\prime)$ is a vertex of $H$ and $g(e^\prime)$ is not, define $h$ by $h(f(e^\prime))=v$, $h(g(e^\prime))=\ell$ and $h(y)=v$ for all $y\in P(H)\backslash \{ f(e^\prime), g(e^\prime)\}$. Hence $hf\neq hg$. If $f(e^\prime)$ is an edge of $H$, then define $h$ by $h(f(e^\prime))=\ell$ and $h(y)=v$ for all $y\in P(H)\backslash\{ f(e^\prime)\}$. Hence $hf\neq hg$, and $A$ is a cogenerator.\\
\indent If $C$ is a graph not containing any loops, then no morphism exists from $K_1^\ell$ to $C$ that can distinguish between $1_{K_1^\ell},f:K_1^\ell\rightarrow K_1^\ell$, where $1_{K_1^\ell}$ is the identity morphism, and $f$ is the morphism that maps the loop and vertex of $K_1^\ell$ to the vertex of $K_1^\ell$. If $C$ is a graph not containing any non-loop edges, then no morphism exists from $K_2$ to $C$ that can distinguish between $j,k:K_1\rightarrow K_2$ where $j$ maps the single vertex of $K_1$ to one vertex of $K_2$, and $k$ maps the single vertex of $K_1$ to the other vertex of $K_2$. Hence all cogenerators require a loop and a non-loop edge. A similar proof applies to \textbf{SiGrphs}.\\
\indent \textbf{Part 2}: Let $A$ be a graph with an edge $e$ incident to vertices $u_1$ and $u_2$. As in the proof for generators in \textbf{SiLlGrphs}, if $f,g:G\rightarrow H$ are distinct, then there is a vertex $v\in V(G)$ such that $f(v)\neq g(v)$. Define $h:H\rightarrow A$ by $h(f(v))=u_2$, $h(y)=u_1$ for all $y\in V(H)\backslash \{f(v)\}$, $h(a)=e$ for all edges of $H$ incident to $f(v)$, and $h(a)=u_1$ for all other edges in $H$. Hence $hf\neq hg$, and $A$ is a cogenerator.\\
\indent If $C$ is a graph not containing any edges, then no morphism exists from $K_2$ to $C$ that can distinguish between $j,k:K_1\rightarrow K_2$ where $j$ maps the single vertex of $K_1$ to one vertex of $K_2$, and $k$ maps the single vertex of $K_1$ to the other vertex of $K_2$. Hence all cogenerators in \textbf{SiLlGrphs} require an edge.\\
\indent \textbf{Part 3}: Let $A$ be a graph composed of two disconnected components. One component contains a vertex $v$ with two loops $\ell_1$ and $\ell_2$, and the other component contains two vertices $u_1$ and $u_2$ with an edge $e$ between them, and a loop $\ell_{u1}$ and $\ell_{u2}$ on $u_1$ and $u_2$ respectively. If $A$ is a cogenerator, then any graph containing a loop and a non-loop edge is also a cogenerator.\\
\indent Let $f,g:G\rightarrow H$ be two distinct morphisms in \textbf{StGrphs}. Then $f(e^\prime)\neq g(e^\prime)$ for some $e^\prime\in P(G)$. If $f(e^\prime)$ is a vertex, define $h:H\rightarrow A$ by $h(f(e^\prime))=u_2$, $h(y)=u_1$ for all $y\in V(H)\backslash\{ f(e^\prime)\}$, $h(a)=e$ for all non-loop edges $a$ incident to $h(e^\prime)$ in $H$, $h(a)=\ell_{u2}$ for all loops $a$ incident to $h(e^\prime)$, and $h(a)=\ell_{u1}$ for all other edges $a$ of $H$. Hence $hf\neq hg$.\\
\indent If $f(e^\prime)$ is an edge of $H$, then define $h$ by  $h(f(e^\prime))=\ell_1$ and $h(y)=v$ for all $y\in V(H)$ and $h(a)=\ell_2$ for all other edges $a$ in $H$. Hence $hf\neq hg$, and $A$ is a cogenerator.\\
\indent Now suppose $C$ is a graph that has no vertex that is incident to two loops. Then consider $f,g:K_1^\ell\rightarrow B$, where $B$ is a graph composed of one vertex with two loops $\ell_a$ and $\ell_b$, $f$ maps the loop of $K_1^\ell$ to $\ell_a$ and $g$ maps the loop to $\ell_b$. No morphism exists from $B$ to $C$ that can distinguish between $f$ and $g$.\\
\indent Now suppose $C$ is a graph that does not contain any subgraph isomorphic to $K_2^\ell$, then there is no edge incident two any two vertices that have loops. Then there exists no morphism from $K_2^\ell\rightarrow C$ that can distinguish between $j,k:K_1^\ell\rightarrow K_2^\ell$, where $j$ maps the vertex and loop of $K_1^\ell$ to one of the vertices of $K_2^\ell$ and its incident loop, and $k$ maps the vertex and loop of $K_1^\ell$ to the other vertex of $K_2^\ell$ and its incident loop. Hence all cogenerators in \textbf{StGrphs} require a vertex with two loops, and a subgraph isomorphic to $K_2^\ell$.\\
\indent \textbf{Part 4}: Let $K_2^\ell$ have vertices $u$ and $v$ with edge $e$ incident to $u$ and $v$ and loops $\ell_u$ and $\ell_v$ on $u$ and $v$ respectively. Let $X$ and $Y$ be graphs in 
\textbf{SiStGrphs} with morphisms $f,g:X\rightarrow Y$ such that $f\neq g$. Since there is at most one loop at a vertex and at most one edge between any two vertices, there is a vertex $x\in V(X)$ such that $f(x)\neq g(x)$. Define a map $h:Y\rightarrow K_2^\ell$ by $h(f(x))=u$ and $h(y)=v$ for all vertices $y\in V(Y)\backslash \lbrace f(x)\rbrace$, and for $a\in E(Y)$, $h(a)=\ell_v$ if $\partial_Y(a)=(y_1$\textunderscore$y_2)$ for $y_1,y_2\in V(Y)\backslash \lbrace f(x)\rbrace$, $h(a)=\ell_u$ if $\partial_Y(a)=(f(x)$\textunderscore$f(x))$, and $h(a)=e$ if $\partial_Y(a)=(f(x)$\textunderscore$y)$ for $y\in V(Y)\backslash \lbrace f(x)\rbrace$.\\
\indent We now show $h$ is a strict graph homomorphism. Let $a\in E(Y)$. If $h(a)=\ell_v$ then $\partial_{K_2^\ell}(h(a))=(v$\textunderscore$v)=(h(y_1)$\textunderscore$h(y_2))$ for some $y_1,y_2\in V(Y)\backslash \lbrace f(x)\rbrace$. If $h(a)=\ell_u$ then $\partial_{K_2^\ell}(a)=(u$\textunderscore$u)=(h(f(x))$\textunderscore$h(f(x)))$. If $h(a)=e$ then $\partial_{K_2^\ell}(a)=(u$\textunderscore$v)=(h(f(x))$\textunderscore$h(y))$ for some $y\in V(Y)\backslash \lbrace f(x)\rbrace$. Hence $h$ preserves incidence, and since $h$ sends edges to edges; $h$ is a strict graph homomorphism.\\
\indent Since $f(x)\neq g(x)$, $h(f(x))=u$ and $h(g(x))=v$. Hence $h f\neq h g$, and $K_2^\ell$ is a cogenerator of \textbf{SiStGrphs}.\\
\indent If $G$ is a graph in \textbf{SiStGrphs}  that contains a subgraph isomorphic to $K_2^\ell$ then clearly $(i h) f\neq (i h) g$, where $i$ is the inclusion morphism (over the isomorphism) $i:K_2^\ell\rightarrow G$. Hence $G$ is a cogenerator of \textbf{SiStGrphs}.\\
\indent Suppose $C$ does not have a subgraph isomorphic to $K_2^\ell$ but $C$ is a cogenerator. Then no two vertices of $G$ with loops are incident to the same edge. Consider the two morphisms $1_{K_2^\ell},t_w:K_2^\ell\rightarrow K_2^\ell$ where $1_{K_2^\ell}$ is the identity morphism and $t_w$ is the morphism defined by $t_w(u)=v$, $t_w(v)=u$, $t_w(e)=e$, $t_w(\ell_u)=\ell_v$, and $t_w(\ell_v)=\ell_u$. Since $C$ is a cogenerator, there is a morphism $h:K_2^\ell\rightarrow C$ such that $h 1_{K_2^\ell} \neq h t_w$.\\
\indent Let $h(u)=u^\prime$ for some $u^\prime \in V(C)$ and $h(v)=v^\prime$ for some $v^\prime \in V(C)$. If $u^\prime=v^\prime$, then since edges must be sent to edges and incidence is preserved, $\partial_C(h(1_{K_2^\ell}(\ell_v)))=(h(v)$\textunderscore$h(v))=(v^\prime$\textunderscore$v^\prime)=
(u^\prime$\textunderscore$u^\prime)=\partial_C(h(t_w(\ell_v)))$. Since there is at most one loop at a vertex, then $h(1_{K_2^\ell}(\ell_v))=h(t_w(\ell_v))$. Similarly $h(1_{K_2^\ell}(\ell_u))=h(t_w(\ell_u))$ and $h(1_{K_2^\ell}(e))=h(t_w(e))$. Hence $h 1_{K_2^\ell} = h t_w$, a contradiction. Thus $u^\prime \neq v^\prime$.\\
\indent Since morphisms must send edges to edges, $\partial_C(h(\ell_u))=(u^\prime$\textunderscore$u^\prime)$, and $\partial_C(h(\ell_v))=(v^\prime$\textunderscore$v^\prime)$, $u^\prime$ has a loop $\ell_{u^\prime}$ and $v^\prime$ has a loop $\ell_{v^\prime}$. Now consider $h(e)$. Since $\partial_C(h(e))=(h(u)$\textunderscore$h(v))=(u^\prime$\textunderscore$v^\prime)$, $u^\prime$ and $v^\prime$ are two vertices with loops adjacent to the same edge, a contradiction. Hence $C$ must contain a subgraph isomorphic to $K_2^\ell$.\\
\indent \textbf{Part 5}:
To show there are no cogenerators in \textbf{SiLlStGrphs}, we show there is no graph $X$ such that any graph $G$ admits a morphism to $X$. Assume such a graph $X$ exists. Consider the complete graph $K$ such that the cardinailty of $V(K)$ is greater than the cardinality of $V(X)$. By hypothesis there is a morphism $f:K\rightarrow X$. Since $\sharp(V(K))>\sharp(V(X))$ and $\overline{f}_V$ is a set map, there are two distinct vertices $u,v\in V(K)$ such that $\overline{f}(u)=\overline{f}(v)$. Let $e$ be the edge in $K$ incident to both $u$ and $v$. Since graph homomorphisms preserve incidence and $\overline{f}(u)=\overline{f}(v)$, $\partial_X(\overline{f}(e))=(\overline{f}(u)$\textunderscore$\overline{f}(u))$. Then since edges must be sent to edges, $\overline{f}(e)$ is a loop. This contradicts $X$ being loopless. Hence no such object exists.\\
\indent We now note that every graph $G$ with at least two vertices admits at least two distinct morphisms to the complete graph $K$ with the cardinality of $V(K)$ equal to the cardinality of $V(G)$, as the automorphism group of $K$ is the symmetric group on $V(K)$. This fact coupled with the fact there is no graph such that any other graph admits a morphism to it proves no cogenerators exist.
\end{proof}
\indent We will now provide a reference table.
\begin{table}[h]
\caption{Categorial Properties in the categories of graphs.}
\begin{tabular}{ r || l | l | l | l | l | l | l | }			
  \quad & \textbf{SiLlStG} & \textbf{SiLlG} & \textbf{SiStG} & \textbf{SiGrphs} & \textbf{StG} & \textbf{G} \\
\hline
  Epis are surj. on &  vert. sets & vert. sets & vert. sets & part sets & part sets & part sets \\
\hline  
  Monos are inj. on  &  vert. sets & vert. sets & vert. sets & part sets & part sets & part sets \\
\hline 
  Free Graphs & $K_n^c$ & $K_n^c$ & $K_n^c$ & $K_n^c$ & $K_n^c$ & $K_n^c$ \\
\hline 
  Cofree Graphs &  $\nexists$ & $K_n$ & $K_n^\ell$ & $K_n$ & $K_n^\ell$ & $K_n$ \\
\hline 
  Projective Objects & Y & Y & Y & Y & Y & Y  \\
\hline 
  Enough Proj. & Y & Y & Y & Y & Y & Y \\
\hline 
  Injective Objects  & N & Y & Y & Y & Y & Y \\
\hline 
  Enough Inj. & N & Y & Y & Y & Y & Y \\
\hline 
  Generators & Y & Y & Y & Y & N & Y\\
\hline
  Cogenerators & N & Y & Y & Y & Y & Y\\
\hline
\end{tabular}
\end{table}

\section{Conclusion - Making Informed Categorial Choices}
\indent When using category theory in the study of graph morphisms, a choice must be made of which category of graphs to use. An informed choice of category can be based on what categorial properties would be useful to the study.\\
\indent From only the Lawvere-type Axioms, the six categories of graphs break up into three types. The first types consists of the ``smallest'' category in terms of morphisms and objects: \textbf{SiLlStGrphs}. It fails to satisfy any of the Lawvere-type Axioms. The other five categories fall into the other two types and they have limits and a natural number object, but they fail to have choice and are not two-valued.\\
\indent The second type consists of \textbf{SiStGrphs} and \textbf{SiLlGrphs}. These categories have exponentiation with evaluation, but they lack a subobject classifier. The last type consists of \textbf{SiGrphs}, \textbf{StGrphs} and \textbf{Grphs}, each of which have a subobject classifier but they do not have exponentiation with evaluation. We note that all six categories of graphs fail to be topoi.\\
\indent However, within these 3 types, there are vast differences when we look at the other categorial properties. The terminal object, free objects, and cofree objects depend on the type of morphism chosen for the category, highlighting a difference between the categories in our second type, \textbf{SiLlGrphs} and \textbf{SiStGrphs}. This property also differentiates \textbf{StGrphs} from \textbf{SiGrphs} and \textbf{Grphs}. Another categorial property that singles out \textbf{StGrphs} from the third type is a lack of generators.\\
\indent \textbf{Grphs} and \textbf{SiGrphs} have the same categorial properties investigated in this paper, but they can satisfy constructions in different ways. For example, consider the coequalizer for the two morphisms $f,g:K_2\rightarrow K_2^{\ell}$ both of which send the edge of $K_2$ to the non-loop edge of $K_2^{\ell}$ but differ on where they send the two incident vertices. In \textbf{Grphs} the coequalizer object will be a graph with one vertex and two loops, whereas in \textbf{SiGrphs} the coequalizer object will be isomorphic to $K_1^\ell$, as the two loops must be identified. When considering categorial quotients, this difference in construction can be crucial.\\

\addcontentsline{toc}{section}{References}
\bibliographystyle{alpha}
\bibliography{concgrphs}
\noindent \textsc{Department of Mathematical Sciences, The University of Montana} \\
\textsc{Missoula, MT 59812-0864, USA} \\
\textit{E-mail}: george.mcrae@umontana.edu\\
\indent \qquad demitri.plessas@umontana.edu\\
\indent \qquad liam.rafferty@umontana.edu\\

\end{document}